\documentclass[a4paper]{amsart}

\usepackage{amsmath,amsthm,amssymb}
\usepackage[pdftex]{color,graphicx}
\usepackage[latin1]{inputenc}

\usepackage[colorlinks,linkcolor=blue]{hyperref}
\usepackage[nameinlink]{cleveref}

\usepackage{bm}
\usepackage{amsfonts}
\usepackage{geometry} 
\usepackage[onehalfspacing]{setspace}
\usepackage{graphicx}
\usepackage{amsthm,amsmath}
\usepackage[dvipsnames]{xcolor}
\usepackage{ulem}
\usepackage{tikz}
\usepackage{enumerate}
\usepackage{amsmath}
\usepackage{mathdots}
\usepackage{yhmath}
\usepackage{cancel}
\usepackage{color}
\usepackage{siunitx}
\usepackage{array}
\usepackage{multirow}
\usepackage{amssymb}
\usepackage{tabularx}
\usepackage{extarrows}
\usepackage{booktabs}
\usepackage{soul}
\usepackage{mathpazo}
\usepackage{hyperref}
\usepackage{multimedia}
\usepackage{rotating}
\usetikzlibrary{fadings}
\usetikzlibrary{patterns}
\usetikzlibrary{shadows.blur}
\usetikzlibrary{shapes}
\usepackage{float}
\usetikzlibrary{positioning}
\graphicspath{{/figs/}}
\usepackage{xcolor}
\usepackage{datetime}
\newdateformat{monthyeardate}{%
  \monthname[\THEMONTH], \THEYEAR}

\newtheorem{theorem}{Theorem}
\newtheorem*{theorem*}{Theorem}

\newtheorem{corollary}{Corollary}

\newtheorem{definition}{Definition}

\newtheorem{lemma}{Lemma}

\newtheorem{remark}{Remark}

\usetikzlibrary{arrows.meta}
\usepackage{thmtools}

\declaretheoremstyle[headfont=\bf]{normalhead}
\declaretheorem[style=normalhead]{example}

\usepackage[none]{hyphenat}
\usepackage{tikz-cd}
\usetikzlibrary{decorations.markings}

\tikzset{degil/.style={
            decoration={markings,
            mark= at position 0.5 with {
                  \node[transform shape] (tempnode) {$\bm{\times}$};
                  }
              },
              postaction={decorate}
}
}

\begin{document}

\title[local non-bossiness]{local non-bossiness}\thanks{We are indebted to Keisuke Bando, Haluk Ergin, William Thomson, the Editor Faruk Gul, and two anonymous referees for their suggestions. We thank (in alphabetic order) Agust\'in Bonifacio, Julien Combe, Gabriela Denis, Juan Dubra, Federico Echenique, Lars Ehlers, Onur Kesten, SangMok Lee, Jordi Mass\'o, Ran Shorrer, Camilo J. Sirguiado, Can Urgun and Bumin Yenmez for their valuable comments and suggestions, as well as the conference/seminar participants at the 35th Stony Brook International Conference on Game Theory, LACEA-LAMES 2024, Universidad de Montevideo, London School of Economics, Universidad de Salamanca, Universit\'e de Montr\'eal, Universidad Nacional de San Luis, SEU meeting 2024, and the 14th Conference on Economic Design. Financial support from ANII, FCE\_1\_2023\_1\_176072, is gratefully acknowledged.}

\author[Duque-Rosas]{Eduardo Duque-Rosas}
\thanks{E. Duque-Rosas, London School of Economics and Political Science (e.g.duque@lse.ac.uk)}
\author[Pereyra]{\,Juan S. Pereyra}
\thanks{J. S. Pereyra, Universidad de Montevideo, Uruguay (jspereyra@um.edu.uy)}
\author[Torres-Mart\'{\i}nez]{\,Juan Pablo Torres-Mart\'{\i}nez}
    \thanks{J. P. Torres-Mart\'{\i}nez, Department of Economics, Faculty of Economics and Business, University of Chile (juan.torres@fen.uchile.cl)}
\date{\monthyeardate\today}

\thispagestyle{empty}

\begin{abstract}
The student-optimal stable mechanism (DA), the most popular mechanism in school choice, is the only one that is stable and strategy-proof. However, when DA is implemented, a student can change the schools of others without changing her own. We show that this drawback is limited: a student cannot change her schoolmates while remaining at the same school. We refer to this new property as {\it{local non-bossiness}} and use it to provide a new characterization of DA that does not rely on stability. Furthermore, we show that local non-bossiness plays a crucial role in providing incentives to be truthful when students have preferences over their colleagues. As long as students first consider the school to which they are assigned and then their schoolmates, DA induces the only stable and strategy-proof mechanism. There is limited room to expand this preference domain without compromising the existence of a stable and strategy-proof mechanism.\\
\medskip



\end{abstract}

\maketitle

\thispagestyle{empty}
\newpage

\section{Introduction}

In the last two decades, an increasing number of centralized school admission systems have been implemented worldwide.\footnote{The website www.ccas-project.org provides information about school choice systems in different countries.}  In this context, the {\it{student-optimal stable mechanism}} (${\rm{DA}}$) of Gale and Shapley (1962) has become one of the most popular mechanisms for distributing school seats. These developments are rooted in the theoretical results laid out in the seminal works of Roth (1985), Roth and Sotomayor (1989), Balinski and S\"{o}nmez (1999), Abdulkadiro\u{g}lu and S\"{o}nmez (2003), and Pathak and S\"{o}nmez (2013).\footnote{We refer to the surveys of Abdulkadiro\u{g}lu (2013), Kojima (2017), and Pathak (2017) for a discussion of recent developments in school choice and their applications.} 

The success of ${\rm{DA}}$ might be explained by its properties. When students have strict preferences over schools, it is well-known that ${\rm{DA}}$ is the only stable and strategy-proof mechanism (Dubins and Freedman, 1981; Roth, 1982; Alcalde and Barber\`a, 1994). Stability ensures that school seats are not wasted and that no one would prefer the seat of someone with lower priority for it (i.e., there is no ``justified envy''). Strategy-proofness guarantees that no student has incentives to misreport her preferences.

However, when ${\rm{DA}}$ is implemented, a change in a student's preference can modify the assignment of others without changing her own (Roth, 1982).  
In technical terms, ${\rm{DA}}$ is {\it{bossy}}, and this is important for at least two reasons. 
First, it is related to the Pareto inefficiency of ${\rm{DA}}$ and to the fact that a coalition of students could improve their assignments by jointly misrepresenting their preferences (Pap\'ai, 2000; Ergin, 2002). Second, it has significant consequences for incentives when students care not only about the school to which they are assigned but also about the assignment of others. Indeed, it is the bossiness of ${\rm{DA}}$ that compromises the existence of a stable and strategy-proof mechanism when there are externalities but each student prioritizes her own school (Duque and Torres-Mart\'{\i}nez, 2023).

In this paper, we introduce a new incentive property satisfied by ${\rm{DA}}$ that limits its bossiness. The new property, called {\it{local non-bossiness}}, says that a student cannot change her schoolmates without changing the school to which she is assigned.  
Thus, when a locally non-bossy mechanism is used, a student's capacity to modify others' assignment without changing her school is limited to those who are not assigned to the same school. This includes the case where the student is not assigned: she cannot modify the set of unassigned students while remaining without a school. Local non-bossiness is therefore non-trivial even in a one-to-one matching problem.
We use local non-bossiness, along with a set of axioms that do not include stability, to characterize the DA mechanism. 
We also study a model with externalities and show that local non-bossiness ensures the existence of a stable and strategy-proof mechanism when students care first about the school to which they are assigned and then about their schoolmates. 

Our analysis begins by showing that ${\rm{DA}}$ is locally non-bossy (\autoref{loc_nonb}) and that this property is independent of both stability and strategy-proofness. Furthermore, every locally non-bossy and strategy-proof mechanism is {\it{locally group strategy-proof}}, in the sense that no coalition of {\it{schoolmates}} can manipulate it to improve the situation of at least one of its members without harming the others. In particular, ${\rm{DA}}$ is the only stable and locally group strategy-proof mechanism (\autoref{cor1}). We also study the relationships between local non-bossiness, local group strategy-proofness, and other incentive properties (\autoref{diag} summarizes our results).

Secondly, we use local non-bossiness to provide a new characterization of DA that does not involve stability. We fix a set of schools and capacities, and consider a mechanism as a function that assigns a matching to each set of students and their preferences. We show that a mechanism satisfies individual rationality, weak non-wastefulness,
population-monotonicity, strategy-proofness, S-WrARP, and weak local non-bossiness if and only if it is the DA mechanism for some profile of priorities (\autoref{coro-bando}). Under individual rationality, no student is assigned to a school that she considers unacceptable; weak non-wastefulness ensures that unassigned students never prefer a school that did not fill its places; and population-monotonicity guarantees that no one gets worse off when the set of students shrinks. Along with strategy-proofness, these three axioms characterize DA when schools have only one seat available (Ehlers and Klaus, 2016).\footnote{We thank Keisuke Bando for pointing out the relation between local non-bossiness and the characterization of Ehlers and Klaus (2016) for the unit capacity case.} S-WrARP is a restricted version of the {\it{weaker axiom of revealed preference}} introduced by Jamison and Lau (1973), but applied to schools' choice functions induced by a mechanism. It is a necessary condition to 
ensure that these choice functions are consistent with priority orders and finite capacities (Chambers and Yenmez, 2018). Weak local non-bossiness restricts local non-bossiness to schools, allowing a student without a school to modify the set of unassigned students without getting a seat at any school. S-WrARP and weak local non-bossiness make it possible to manage schools with more than one seat available, and they are trivially satisfied when each school has only one seat available. Strengthening strategy-proofness to group strategy-proofness, our axiomatization allows us to characterize DA when schools' priority orders are acyclic in the sense of Ergin (2002), a context where ${\rm{DA}}$ is Pareto efficient (\autoref{coro-bando2}). 

Thirdly, we introduce externalities in our model by allowing students to have preferences over the set of matchings. In particular, we analyze the implications of the local non-bossiness of ${\rm{DA}}$ in this context. It is well known that many of the results in the literature break down in the presence of externalities, for example, the existence of a stable matching. One way to recover this result is to restrict preferences to be \textit{school-lexicographic}. That is, by assuming that each student is primarily concerned with her assigned school, and when assigned to the same school in two different matchings, there is no restriction on how to compare them (Sasaki and Toda, 1996; Dutta and Mass\'o, 1997; Fonseca-Mairena and Triossi, 2023). 

However, even on this restricted domain, a stable and strategy-proof mechanism may not exist (Duque and Torres-Mart\'inez, 2023). 
We further restrict the preference domain to the family of \textit{school-lexicographic preferences over colleagues}, assuming that every student cares first about the school she is assigned to and then only about her schoolmates. That is, each student is indifferent among all the matchings in which she is assigned to the same school with the same colleagues. There are no restrictions on the order in which she ranks two matchings that assign her to the same school but with different schoolmates. In this context, we demonstrate that applying ${\rm{DA}}$ to the school rankings underlying school-lexicographic preferences over colleagues induces a stable and strategy-proof mechanism (\autoref{propo-ex}). Intuitively, a student may want to misreport her preferences either to change her school or to maintain it and change her schoolmates. The first reason for misreporting preferences is already present in classical school choice problems, and avoiding it relates to ensuring strategy-proofness. The second one emerges in the presence of school-lexicographic preferences, and avoiding it relates to guaranteeing local non-bossiness. Therefore, the local non-bossiness and strategy-proofness of ${\rm{DA}}$ are key to ensuring the existence of a stable and strategy-proof mechanism in this context.

Consequently, the incompatibility between stability and strategy-proofness in contexts where students prioritize their own school is not caused by the existence of preferences over the assignment of others {\it{per se}}, but by the fact that these preferences extend beyond their schoolmates. Evidently, it is reasonable to maintain the dependence of preferences on schoolmates because it is a well-documented empirical phenomenon in school choice (Rothstein, 2006; Abdulkadiro\u{g}lu et al., 2020; Allende, 2021; Che et al., 2022; Beuermann et al., 2023; Cox et al., 2023).

Our results have practical implications for admission systems based on the ${\rm{DA}}$ mechanism. They show that ${\rm{DA}}$ still performs well when students care about the assignments of others, as long as they first consider their assigned school and then their schoolmates. 
More precisely, when DA is implemented and students are only required to report strict rankings of schools, it is a weakly dominant strategy for each student with school-lexicographic preferences over colleagues to report the school ranking induced by her true preferences. Furthermore, despite the variety of mechanisms that could be defined using information on preferences for schools and schoolmates, ${\rm{DA}}$ induces the only stable and strategy-proof mechanism on this domain (\autoref{unique}). It is important to note that there is limited room to relax the assumption of school-lexicographic preferences over colleagues. Indeed, it is enough for the preferences of just one student to be school-lexicographic but not school-lexicographic over colleagues to prevent the existence of a stable and strategy-proof mechanism (see \autoref{ext}).\\

\noindent {\bf{Related literature.}} We contribute to three strands of the literature: the analysis of the bossiness of ${\rm{DA}}$, the axiomatization of DA without appealing to stability, and the study of school choice problems with externalities. 

The concept of non-bossiness was first introduced by Satterthwaite and Sonnenschein (1981) and has been studied extensively in the context of the assignment of indivisible goods.\footnote{We refer to the critical survey by Thomson (2016) for a detailed discussion of the multiple interpretations of non-bossiness and its implications.} 
Non-bossiness plays an essential role in avoiding coalitional manipulability, as group strategy-proofness is equivalent to the combination of strategy-proofness and non-bossiness when preferences are strict (Pap\'ai, 2000).  Moreover, for a given pair of schools' priorities and capacities, Ergin (2002) shows that ${\rm{DA}}$ is non-bossy if and only if it is Pareto efficient. Kojima (2010) demonstrates the incompatibility between non-bossiness and stability in college admission problems, a scenario where both sides of the market may act strategically. However, when only one side of the market reports preferences, as in school choice problems, Afacan and Dur (2017) show that the school-optimal stable mechanism is non-bossy for students. We contribute to this strand of the literature by highlighting that the bossiness of ${\rm{DA}}$ is limited: when it is implemented, no student can change her schoolmates without also changing her assigned school. We also show that this property implies that no group of students assigned to the same school can manipulate DA to either improve the assignment of at least one of them without hurting another member (local group strategy-proofness), or maintain the school and change their schoolmates (local group non-bossiness).

Recently, Raghavan (2020) reframes bossiness in terms of an ``influence'' relation: an agent's ability to affect others' welfare without changing her own. He characterizes which agents have influence and whether the influenced agents are better off or worse off in welfare terms. From this perspective, we study a particular set of potentially influenced agents under {\rm{DA}}---namely, the students assigned to the same school as the student who misreports her preferences---and show that they cannot be influenced. 

The DA mechanism has been axiomatized in several ways. It is the only stable mechanism that is weakly Pareto efficient (Roth, 1982), and the only stable and strategy-proof mechanism (Alcalde and Barber\`a, 1994).\footnote{When students always report all schools as admissible, DA is not necessarily the only stable and strategy-proof mechanism (Sirguiado and Torres-Mart\'{\i}nez, 2024).}
Balinski and S\"{o}nmez (1999) and Kojima and Manea (2010) provide alternative axiomatizations without appealing to strategy-proofness, while Kojima and Manea (2010), Morrill (2013), and Ehlers and Klaus (2014, 2016) characterize DA without appealing to stability. In particular, assuming that  each school has only one seat available, Ehlers and Klaus (2016) show that a mechanism satisfies individual rationality, weak non-wastefulness, population-monotonicity, and strategy-proofness if and only if it is DA for some priority profile. However, they show that this characterization does not hold in a general capacity model. By introducing two new axioms, S-WrARP and weak local non-bossiness, we extend Ehlers and Klaus' (2016, Theorem 1) axiomatization of DA to many-to-one school choice problems.

In matching problems with externalities, an agent cares not only about her match but also about the matches of others. This literature was initiated by Sasaki and Toda (1996), who studied stability concepts in marriage markets. In recent years, several authors have extended the analysis to many-to-one matching models and other allocation problems. They have introduced restrictions on preference domains to specify the type of externalities and to ensure the existence of stable outcomes.\footnote{The works of Bando, Kawasaki, and Muto (2016) and Pycia and Yenmez (2023) provide excellent surveys on the evolution of this literature.} In the context of many-to-one matching problems, Dutta and Mass\'o (1997) studied a two-sided model with workers and firms where workers' preferences are lexicographic. When workers' preferences over firms dictate their overall preferences over firm-colleague pairs, they show that the set of stable matchings is non-empty (cf., Fonseca-Mairena and Triossi, 2023). Complementing this approach, Echenique and Yenmez (2007) assume that workers have general preferences over colleagues and present an algorithm that produces a set of allocations containing all stable matchings, if they exist. Moreover, Revilla (2007) and Bykhovskaya (2020) determine subdomains of preferences over colleagues for which a stable matching always exists.  More recently, Pycia and Yenmez (2023) studied a hybrid model that allows for general types of externalities, including school choice problems as a particular case. They demonstrate that a stable matching exists as long as externalities affect students' choice rules in such a way that a {\it{sustitutability condition}} is satisfied, which is also necessary for stability to some extent. 

Although Dutta and Mass\'o (1997) demonstrate that the set of stable matchings when preferences are school-lexicographic is the same as without externalities, Duque and Torres-Mart\'inez (2023) show that a stable and strategy-proof mechanism may not exist. We contribute to this strand of the literature by identifying a new preference domain where this incompatibility does not hold. In particular, when students first consider the school to which they are assigned and then only their colleagues, ${\rm{DA}}$ induces the only stable and strategy-proof mechanism. Moreover, it is the local non-bossiness of ${\rm{DA}}$ that underlies this result.

The remainder of the paper is organized as follows. \autoref{model} describes the classical school choice problem and the concepts of local non-bossiness and local group strategy-proofness. In \autoref{locNB} we show that ${\rm{DA}}$ is locally non-bossy. \autoref{axioma} provides an axiomatic characterization of DA.  In \autoref{ext} we introduce externalities in our model, and we show that ${\rm{DA}}$ induces a stable and strategy-proof mechanism on the domain of school-lexicographic preferences over colleagues. \autoref{conclu} contains some final remarks. Omitted proofs and examples are left to the appendices. 

\section{Model}\label{model}

Let $N$ be a set of students and $S$ a set of schools. Each $s\in S$ has a priority order $\succ_s$ over students, and a capacity $q_s\geq 1$. Every $i\in N$ is characterized by a complete, transitive, and strict preference relation $P_i$ defined on $S\cup \{s_0\}$, where $s_0$ represents an outside option. Denote by $R_i$ the weak preference induced by $P_i$ and refer to a school $s$ as admissible when $sP_i s_0$. Given $\succ=(\succ_s)_{s\in S}$, $q=(q_s)_{s\in S}$, and $P=(P_i)_{i\in N}$, we refer to $[N,S,\succ, q]$ as a \textbf{school choice context} and to $[N,S,\succ, q, P]$ as a \textbf{school choice problem}. 

A matching is a function $\mu:N\rightarrow S\cup\{s_0\}$ that assigns students to schools in such a way that at most $q_s$ students are assigned to $s\in S$. When $\mu(i)=s_0$, student $i$ is not assigned to any school. Let ${\mathcal{M}}$ be the set of matchings and $\mu(s)=\{i\in N: \mu(i)=s\}$ be the set of students that $\mu$ assigns to $s\in S\cup \{s_0\}$. A matching $\mu$ is {\textbf{individually rational}} if no student $i$ prefers $s_0$ to $\mu(i)$, while it is {\bf{non-wasteful}} when there is no $(i,s)\in N\times S$ such that $sP_i \mu(i)$ and $\vert \mu(s)\vert<q_s$. A student $i$ has {\bf{justified envy}} at $\mu$ when there exists $s\in S$ such that $s P_i \mu(i)$ and $i\succ_s j$ for some $j\in \mu(s)$. A matching $\mu$ is {\textbf{stable}} when it is individually rational, non-wasteful, and free from justified envy.\footnote{Equivalently, a matching $\mu$ is stable when it is individually rational and there is no {\bf{blocking pair}} $(i,s)\in N\times S$, in the sense that $sP_i \mu(i)$ and either $\vert \mu(s)\vert <q_s$ or $i\succ_s j$ for some $j\in \mu(s)$.}

A mechanism is a function that associates a matching to each school choice problem. The mechanism designer is assumed to know the school choice context $[N,S,\succ, q]$ and that students' preferences belong to the preference domain ${\mathcal{P}}={\mathcal{L}}^{\vert N\vert}$, where ${\mathcal{L}}$ is the set of strict linear orders defined on $S\cup \{s_0\}$. Given $P=(P_i)_{i\in N}\in {\mathcal{P}}$ and $C\subseteq N$, let $P_C=(P_i)_{i\in C}$ and $P_{-C}=(P_i)_{i\notin C}$.

When a school choice context $[N,S,\succ, q]$ is fixed, we will consider only students' preferences as the argument of a mechanism $\Phi:{\mathcal{P}}\rightarrow {\mathcal{M}}$, that is $\Phi(P)$ instead of $\Phi(N,S,\succ, q, P)$. Denote by $\Phi_i(P_i, P_{-i})$ the assignment of $i$ when she declares preferences $P_i$ and the other students declare $P_{-i}=(P_j)_{j\neq i}$. Similarly, $\Phi_s(P)$ denotes the set of students assigned to $s\in S\cup \{s_0\}$. 

A mechanism $\Phi$ is {\bf{stable}} ({\bf{individually rational}}) if for all $P\in{\mathcal{P}}$ the matching $\Phi(P)$ is stable (individually rational). Consider the following properties:

\begin{itemize}
\item $\Phi$ is {\textbf{strategy-proof}} if there are no $i\in N$, $P\in {\mathcal{P}}$, and $P'_i\in {\mathcal{L}}$ such that $\Phi_i(P'_i, P_{-i}) P_i  \Phi_i(P)$.  
 
\item $\Phi$ is {\textbf{non-bossy}} if for all $i\in N$, $P \in {\mathcal{P}}$, and $P'_i \in {\mathcal{L}}$, $\Phi_i(P) =  \Phi_i(P'_i, P_{-i})$ implies that $\Phi(P) = \Phi(P'_i, P_{-i})$.

\item $\Phi$ is {\textbf{group strategy-proof}} if there are no $P\in {\mathcal{P}}$, $C\subseteq N$, and $P'_C\in {\mathcal{L}}^{\vert C\vert}$ such that:
\begin{itemize}
\item For some $i\in C$, $\Phi_i(P'_{C}, P_{-C}) \, P_i  \,\Phi_i(P)$. 
\item For each $j\in C$, $\Phi_j(P'_{C}, P_{-C}) \,R_j \, \Phi_j(P)$.
\end{itemize}
\end{itemize}

Strategy-proofness requires that no one has incentives to
misreport her preferences, while non-bossiness ensures that a change in the preferences of a student cannot modify the schools of others without changing her own. Under group strategy-proofness no coalition of students can benefit by misrepresenting their preferences.\\

When a {\it{non-bossy}} mechanism is implemented, no one can change the assignment of other students without changing her own.\footnote{We refer to a mechanism as {\it{bossy}} when it is not non-bossy.} If instead of considering the change in the assignment of all students, we focus only on those who are assigned to the same school as the one who changes her preferences, we have a \textit{local version} of non-bossiness.  
\medskip

\begin{definition}\label{loc_nb}
A mechanism $\Phi$ is {\textbf{locally non-bossy}} if for all $i\in N$, $P\in {\mathcal{P}}$, $P'_i\in {\mathcal{L}}$, and $s\in S\cup \{s_0\}$, $\Phi_i(P) =  \Phi_i(P'_i, P_{-i})=s$ implies that $\Phi_s(P) =  \Phi_s(P'_i, P_{-i})$.   \\
\end{definition}

When a mechanism is locally non-bossy, no student can modify her schoolmates without also changing the school to which she is assigned. This definition accommodates the case in which a student remains unassigned after changing her preferences. In this case, the student cannot modify the set of unassigned students without getting a seat in a school. This implies that local non-bossiness is non-trivial even in a one-to-one matching problem.\\

When a mechanism is not group strategy-proof, a coalition of students can benefit by manipulating their preferences. The members of this coalition may originally be assigned to different schools. When we restrict the coalition of students to be originally assigned to the same school, we have a \textit{local version} of group strategy-proofness.\medskip

\begin{definition}
A mechanism $\Phi$ is {\textbf{locally group strategy-proof}} if there are no $s\in S\cup \{s_0\}$, $P\in {\mathcal{P}}$, $C\subseteq \Phi_s(P)$, and $P'_C\in  {\mathcal{L}}^{\vert C\vert}$ such that:
\begin{itemize}
\item For some $i\in C$, $\Phi_i(P'_{C}, P_{-C}) \, P_i  \,\Phi_i(P)$. 
\item For each $j\in C$, $\Phi_j(P'_{C}, P_{-C}) \,R_j \, \Phi_j(P)$. \medskip
\end{itemize}
    \end{definition}

When a mechanism is locally group strategy-proof, no coalition of {\it{schoolmates}} can manipulate it to improve the welfare of at least one of its members. \\

In \autoref{app2} we analyze the relationships between local non-bossiness, local group strategy-proofness, and other incentive properties (\autoref{diag} summarizes our results). 

Specifically, we demonstrate that local non-bossiness is independent of both stability and strategy-proofness, while local group strategy-proofness is weaker than group strategy-proofness and stronger than strategy-proofness. Furthermore, a locally non-bossy and strategy-proof mechanism is locally group strategy-proof (\autoref{papai}), but a locally group strategy-proof mechanism might violate local non-bossiness. This last result contrasts with the well-known equivalence between group strategy-proofness and the combination of non-bossiness and strategy-proofness (Pap\'ai, 2000). Additionally, we  show in \autoref{afacan} 
that a locally non-bossy and strategy-proof mechanism satisfies a local notion of group non-bossiness (Afacan, 2012), and in \autoref{BaIma} that local non-bossiness is independent of weak non-bossiness (Bando and Imamura, 2016).\footnote{A notion related to non-bossiness is {\it{consistency}} (Ergin, 2002). A mechanism is consistent if after removing some students with their seats, the rest of the students remain in the same school when the mechanism is applied to the reduced market. It is easy to formulate a local notion of consistency and to show that a locally consistent mechanism is locally non-bossy.}\medskip\medskip

\section{Local non-bossiness of DA}\label{locNB}\medskip

Given a school choice context $[N,S,\succ, q]$, let ${\rm{DA}}:{\mathcal{P}}\rightarrow {\mathcal{M}}$ be the {\textbf{student-optimal stable mechanism}}, namely the rule that associates to each school choice problem the outcome of the deferred acceptance algorithm when students make proposals.
\medskip 
{\small{
\begin{quote}

\noindent \underline{Student-Proposing Deferred Acceptance Algorithm}
\medskip

\noindent Given $[N,S,(\succ_s,q_s)_{s\in S}, (P_i)_{i\in N}]$, \medskip

\noindent {\it{Step 1:}} Each student $i$ proposes to the most preferred admissible school according to $P_{i}$, if any. Among those who have proposed to it, each school $s$ provisionally accepts the $q_{s}$-highest ranked students according to $\succ_{s}$.\medskip

\noindent {\it{Step $k\geq 2$:}} Each student $i$ who has not been provisionally accepted in step $k-1$ proposes to the most preferred admissible school according to $P_{i}$, among those to which she has not previously proposed, if any. Each school $s$ provisionally accepts the $q_{s}$-highest ranked students according to $\succ_{s}$, among those who have proposed to it in this step or were provisionally accepted by it in step $k-1$.
\medskip

The algorithm terminates at the step in which no proposals are made. Provisional acceptances become definitive then.\\
\end{quote}
}}

For every preference profile $P\in {\mathcal{P}}$, the matching ${\rm{DA}}(P)$ assigns to each student the most preferred alternative in $S\cup \{s_0\}$ that she can reach in a stable matching under $P$ (Gale and Shapley, 1962; Roth, 1985). Moreover, the mechanism ${\rm{DA}}$ is strategy-proof (Dubins and Freedman, 1981; Roth, 1982). 

However, when DA is implemented, a student can modify the assignment of others without affecting her own (Ergin, 2002). Hence, one may wonder if there is a limit to the impact of these changes. In particular, under ${\rm{DA}}$, is it possible for a student to affect the set of students assigned to her school by reporting different preferences, without changing her school ? As the following theorem shows, the answer is negative. \medskip

\begin{theorem}\label{loc_nonb}
The student-optimal stable mechanism is locally non-bossy.\\
\end{theorem}

To get an intuitive idea of the proof of \autoref{loc_nonb}, let $\mu={\rm{DA}}(P)$ and $\mu'={\rm{DA}}(P'_1,P_{-1})$ be different matchings that give the same assignment to student $1$, namely $\overline{s}\equiv\mu(1)=\mu'(1)$.

First, consider the benchmark case in which $P'_1$ is a {\bf{monotonic transformation}} of $P_1$ for $\overline{s}$, in the sense that the set of schools preferred to $\overline{s}$ shrinks when moving from $P_1$ to $P'_1$. This assumption  implies that $\mu$ is stable under $(P'_1,P_{-1})$, because if there were a blocking pair of $\mu$ under $(P'_1, P_{-1})$, it would be a blocking pair under $P$. Hence, as $\mu'$ is the student-optimal stable matching under $(P'_1,P_{-1})$, $\mu'$ Pareto dominates $\mu$ under $P$.\footnote{This argument proves the {\it{positivity}} of  ${\rm{DA}}$ (cf., Raghavan, 2020, Proposition 2): if a student $i$ reports a monotonic transformation of $P_i$ for ${\rm{DA}}_i(P_i, P_{-i})$ such that ${\rm{DA}}_i(P'_i, P_{-i})={\rm{DA}}_i(P_i, P_{-i})$, then those who modify their assignments are better off.} This property and the stability of $\mu$ guarantee that each student $i$ who changes her assignment from $\mu$ to $\mu'$ displaces someone who was assigned to $\mu'(i)$ at $\mu$. These movements induce \textit{$\mu$-improving cycles} $(i_1, \ldots, i_r)$, where each student $i_k$ strictly prefers $\mu(i_{k+1})$ to $\mu(i_k)$ under $P_k$ [modulo $r$].  Evidently, the implementation of a $\mu$-improving cycle induces justified envy under $P$, because $\mu$ is the student-optimal stable matching. We use this last property to ensure that $\overline{s}$ has the same assignment in $\mu$ and $\mu'$. More precisely, denote by $V$ the set of students who are matched in $\mu$ to schools with different assignments in $\mu$ and $\mu'$. Notice that $V$ includes the students who prefer $\mu'$ to $\mu$. Let $G$ be the graph with nodes $V$ in which every $i$ that prefers $\mu'$ to $\mu$ points to all students who are assigned in $\mu$ to $\mu'(i)$. By allowing students in $V$ to block the changes from $\mu$ to $\mu'$ that induce justified envy, in the proof of \autoref{loc_nonb} we perform an edge replacement in $G$ in order to construct a $\mu$-improving cycle $(i^*_1, \ldots, i^*_r)$ in which justified envy is confined to students outside $V$.\footnote{Specifically, if the change of some student $i$ from $\mu(i)$ to $\mu'(i)$ induces justified envy, we replace every edge from $i$ to a student $j$ such that $\mu(j)=\mu'(i)$ by an edge from $k$ to $j$, where $k$ is the student with justified envy that has the highest priority at $\mu'(i)$. In the proof of Theorem 1, we show that the {\it{multigraph}} obtained by these procedure has a $\mu$-improving cycle that no student in $V$ blocks. }
Since ${\rm{DA}}$ is stable and students in $N \setminus \{1\}$ have the same preferences in $P$ and $(P'_1, P_{-1})$, no student in $(N\setminus V)\setminus \{1\}$ blocks $(i^*_1, \ldots, i^*_r)$. Therefore, student $1$ must block $(i^*_1, \ldots, i^*_r)$, which implies that she does not belong to $V$. We conclude that $\mu(\overline{s})=\mu'(\overline{s})$.

For an arbitrary $P_1'$, let $P^*_1$ be a preference such that $\bar s$ is the best alternative. Given that $1$ is assigned to $\bar s$ in $\mu$ and $\mu'$, we have that $\mu(1)={\rm{DA}}_1(P^*_1,P_{-1})$ and $\mu'(1)={\rm{DA}}_1(P^*_1,P_{-1})$. 
Then, by the first argument, $\mu(\overline{s})={\rm{DA}}_{\overline{s}}(P^*_1,P_{-1})$ and $\mu'(\overline{s})={\rm{DA}}_{\overline{s}}(P^*_1,P_{-1})$, which implies that $\mu(\overline{s})=\mu'(\overline{s})$.\medskip

\begin{proof}[Proof of \autoref{loc_nonb}]

Fix a school choice context  $[S,N,\succ, q]$ and without loss of generality consider student $1$ and preference profiles  $(P_1,P_{-1}), (P'_1, P_{-1})\in {\mathcal{P}}$ such that $\mu\equiv {\rm{DA}}(P_1,P_{-1})$ and $\mu'\equiv {\rm{DA}}(P'_1,P_{-1})$ differ, but $\overline{s}\equiv \mu(1)$ is equal to $\mu'(1)$. We will prove that $\mu(\overline{s})=\mu'(\overline{s})$. \\

\noindent  {\it{{\bf{Case I.}} The preference relation $P'_1$ is a {\it{monotonic transformation}} of $P_1$ for $\overline{s}$.}} 

In this scenario, as $\{s\in S: s P'_1 \overline{s}\}\subseteq \{s\in S: s P_1 \overline{s}\}$, $\mu$ is stable under $(P'_1, P_{-1})$. Since $\mu'$ is the student-optimal stable matching under $(P'_1, P_{-1})$, $\mu'(i) P_i \mu(i)$ for each $i\in N$ such that $\mu(i)\neq\mu'(i)$. Let $I$ be the non-empty set of students who prefer $\mu'$ to $\mu$. The stability of $\mu$ under $(P_1,P_{-1})$ ensures that each student in $I$ is assigned in $\mu'$ to a school that fills its quota at $\mu$. Hence, when we move from $\mu$ to $\mu'$, each student in $I$ displaces another student.

Let $G=(V,E)$ be the graph with nodes and directed edges$$V= \bigcup_{\substack{s\in S\cup \{s_0\}: \ \mu(s) \neq \mu'(s)}} \mu(s),\quad\quad\quad\quad E=\{[i,j]: i\in I,\, \mu(j)=\mu'(i)\}.$$

Notice that $V$ is non-empty, as $I\subseteq V$. The nodes of $G$ are the students who are assigned in $\mu$ to schools that have different assignments in $\mu$ and $\mu'$. Also, there is an edge between every $i\in I$ and all the students who are assigned to $\mu'(i)$ in $\mu$. 

Consider the following concepts:
\begin{itemize}
\item[-] A {\it{$\mu$-improving cycle}}\footnote{We will also refer to a $\mu$-improving cycle as an {\it{improvement cycle from $\mu$}}.}  is a tuple of different students $(i_1,\ldots, i_{r})$ such that $$\quad \quad \mu(i_{l+1}) P_{i_l} \mu(i_l),\quad \,\,\forall l\in \{1,\ldots, r\} \, \,{\mbox{[modulo $r$]}}. $$
 \item[-] A student {\it{$k$ $\mu$-blocks $[i,j]$}} when $\mu(j) P_k \mu(k)$ and $k \succ_{\mu(j)} i$. 
 
 \item[-] A student {\it{$k$ blocks a $\mu$-improving cycle}} $(i_1,\ldots, i_{r})$ when she $\mu$-blocks $[i_l,i_{l+1}]$ for some $l\in \{1,\ldots, r\}$ [modulo $r$].
\end{itemize}

Let $G'=(V,E')$ be the {\it{multigraph}} obtained from $G$ through the following edge-replacement procedure: substitute each edge $[i,j]\in E$ for which
$$I[i,j]\equiv\{k\in V: \,k \ \ \mbox{$\mu$-blocks} \,\, [i,j]\}\neq \emptyset$$
by the directed edge $[\overline{k},j]$, where $\overline{k}$ is the student with the highest priority at school $\mu(j)$ among those in $I[i,j]$. Notice that, $I[i,j]=I[i,j']$ for all $[i,j],[i,j']\in E$. Moreover, when $[\overline{k},j]\notin E$, its inclusion in $E'$ is feasible because $\overline{k}\in I[i,j]\subseteq  V$. 

We make the following remarks regarding $G'$:

\begin{enumerate}
 \item {\it{The graph $G'$ has a cycle. That is, there exists a tuple of different students  $(i^*_1,\ldots, i^*_r)$ such that $[i^*_l, i^*_{l+1}]\in E'$ for all $l\in \{1,\ldots, r\} \, \,{\mbox{[modulo $r$]}}$.}}

Remember that, when we move from $\mu$ to $\mu'$, every student $i$ such that $\mu(i)\neq\mu'(i)$ displaces someone who was assigned to $\mu'(i)$ at $\mu$. Thus, by construction, all nodes of $G$ have a positive in-degree. The same property holds in $G'$. Indeed, when $E'$ is constructed from $E$, if some directed edge $[i,j]$ is deleted, then an edge that points to $j$ is included. Therefore, starting from a node in $V$ and moving backward through the edges in $E'$ we can find a cycle of $G'$.\medskip

 \item {\it{Every cycle in $G'$ is a $\mu$-improving cycle.}}

Remember that, if a student $k$ $\mu$-blocks $[i,j]\in E$, then $\mu(j) P_k \mu(k)$. Hence, each $[k,j]\in  E'\setminus E$ satisfies $\mu(j) P_{k} \mu(k)$. Since the definition of $E$ ensures that every $[k,j]\in  E'\cap E$ satisfies $\mu(j) P_{k} \mu(k)$, we conclude that each cycle in $G'$ is a $\mu$-improving cycle. \medskip

 \item {\it{A $\mu$-improving cycle in $G'$ cannot be blocked by a student in $V$.}}
 
It is a direct consequence of the construction of the set of edges $E'$.  \medskip
 
 \item {\it{Every $\mu$-improving cycle  $(i^*_1,\ldots, i^*_r)$ in $G'$ is blocked by student $1$.}}
 
Since  $\mu$ is the student-optimal stable matching under $(P_1,P_{-1})$, it follows from the previous remark that 
  $(i^*_1,\ldots, i^*_r)$ must be blocked by some student in $N\setminus V$. Moreover, the stability of {$\mu'$} under $(P'_1, P_{-1})$ guarantees that only students in $I\cup\{1\}$ can block $(i^*_1,\ldots, i^*_r)$. Indeed, students who are not in $I\cup\{1\}$ have the same preferences in $P$ and $(P'_1,P_{-1})$, and stay in the same school under $\mu$ and $\mu'$. If such a student blocks the $\mu$-improving cycle, she will also block $\mu'$ because she will have more priority than one of the incoming students at some school.\footnote{Assume that $h\neq 1$ blocks $(i^*_1,\ldots, i^*_r)$. Then, $h$ $\mu$-blocks some $[k,j]\equiv [i^*_l, i^*_{l+1}]$, where  $l\in \{1,\ldots, r\}$ [modulo $r$]. If $s\equiv \mu(j)$, it follows that $s P_h \mu(h)$ and $h\succ_ s k$. We have two possible scenarios: (i) if $[k,j]\in E$, then $\mu'(k)=s$; (ii) if $[k,j]\in E'\setminus E$, then there exists $i\in I$ such that $k$ $\mu$-blocks $[i,j]\in E$, which ensures that $k\succ_ s i$ and $\mu'(i)=s$. Hence, the transitivity of $\succ_s$ guarantees that $h$ has more priority at $s$ than some student assigned to it at $\mu'$. Therefore, as we assume that $s P_h \mu(h)$, if $\mu'(h)=\mu(h)$, we would conclude that $(h,s)$ blocks $\mu'$.}  As $I\subseteq V$, we conclude that $(i^*_1,\ldots, i^*_r)$ is blocked by student $1$.
\end{enumerate}

It follows from remarks (3)-(4) that $1\notin V$, which implies that $\mu(\overline{s})= \mu'(\overline{s})$.\footnote{The construction of $G'$ is not superfluous, because the $\mu$-improving cycle $(i^*_1,\ldots, i^*_r)$, that allows us to show that $\mu(\overline{s})=\mu'(\overline{s}),$ may not be present in the graph $G$ (see \autoref{ex-cycles}).} 

$\,$

\noindent {\it{{\bf{Case II.}} The preference relation $P'_1$ is arbitrary.}}

Let $P^*_1\in {\mathcal{L}}$ be such that $\bar s$ is the best alternative (it is possible that $P^*_1=P_1$). Given that $\mu$ is stable under $(P^*_1,P_{-1})$, and ${\rm{DA}}(P^*_1, P_{-1})$ is the student optimal stable matching under $(P^*_1, P_{-1})$, we have that ${\rm{DA}}_1(P^*_1, P_{-1})=\overline{s}$. Then, $\mu(1) = \mu'(1)={\rm{DA}}_1(P^*_1, P_{-1})=\overline{s}$. It follows from Case I that $\mu(\bar s)={\rm{DA}}_{\overline{s}}(P^*_1,P_{-1})$, and that $\mu'(\bar s)={\rm{DA}}_{\overline{s}}(P^*_1,P_{-1})$. Therefore $\mu(\bar s)=\mu'(\bar s)$.
\end{proof}

\bigskip

The proof of \autoref{loc_nonb} uses the concept of an {\bf{improvement cycle}}: an ordered list of students who exchange schools in a way that makes all of them better off. This concept was previously studied by Kesten (2010) to analyze efficiency adjustments to DA. However, the improvement cycle constructed in the proof of \autoref{loc_nonb}---the one that is blocked only by the student who misreports her preferences---does not necessarily coincide with one of the improvement cycles implemented by the algorithm proposed by Kesten (2010) for mitigating inefficiencies (see \autoref{ex-cycles} and the subsequent discussion). This is not an issue, since our goal is not to address DA's inefficiencies. \medskip

As \autoref{loc_nonb} demonstrates, when ${\rm{DA}}$ is implemented, no student can modify her preferences to change her schoolmates without changing her school. But what happens when multiple students, {\it{all assigned to the same school}}, change their preferences? Can they change their colleagues without changing their school? The answer is negative, and not only for ${\rm{DA}}$, but also for every locally non-bossy and strategy-proof mechanism (see \autoref{afacan}).\medskip

\begin{remark}\label{r1} 
{\textnormal{Given $[N,S,\succ,q]$, when DA is implemented, a student who misreports her preferences and changes her assignment, also modifies all of her colleagues. 
More formally, if ${\rm{Coll}}_i(P)=\{j\in N:  j\neq i, \,\,{\rm{DA}}_j(P)={\rm{DA}}_i(P)\}$ is the set of colleagues of a student $i$ in ${\rm{DA}}(P)$, then the following property holds (see \autoref{DA3}): 
\begin{itemize}
\item For all $i\in N$, $P\in {\mathcal{P}}$, and $P'_i\in {\mathcal{L}}$, 
$${\rm{DA}}_i(P)\neq {\rm{DA}}_i(P'_i, P_{-i})\quad\Longrightarrow \quad {\rm{Coll}}_i(P) \cap {\rm{Coll}}_i(P'_i, P_{-i})=\emptyset.$$ 
\end{itemize}}}

{\textnormal{This property implies that ${\rm{DA}}$ satisfies the converse of local non-bossiness: if the colleagues of some student do not change when her preferences are modified, then her school does not change either. \hfill $\Box$}} \medskip
\end{remark}

\begin{remark}\label{ext} {\textnormal{Throughout the paper, we implicitly assume that schools have {\it{responsive preferences}} (Roth and Sotomayor, 1990), in the sense that priority orders $(\succ_s)_{s\in S}$ and capacities $(q_s)_{s\in S}$ determine the relevant information about how schools choose students from a pool of applicants. This assumption is essential to ensure our main result. Indeed, in \autoref{extension} we provide examples illustrating that the local non-bossiness of ${\rm{DA}}$  cannot be guaranteed when schools' preferences are non-responsive, even in contexts in which ${\rm{DA}}$ is well-defined, stable, strategy-proof, and optimal for students  (cf., Roth, 1984b; Hatfield and Milgrom, 2005).  \hfill $\Box$ }} \medskip
\end{remark}

The ${\rm{DA}}$ mechanism is not the only one that is locally non-bossy and stable in ${\mathcal{P}}$. Indeed, the {\it{school-optimal}} stable mechanism also satisfies these properties, because it is non-bossy (Afacan and Dur, 2017). Furthermore, the arguments made in the proof of Theorem 1 can be easily adapted to provide a direct proof that the school-optimal stable mechanism is locally non-bossy. Indeed, since this mechanism assigns each student to the worst school she can reach in a stable matching, it is sufficient to repeat the proof focusing on {\it{$\mu$-worsening cycles}}: tuples $(i_1,\ldots, i_r)$ of different students such that $\mu(i_l) P_{i_l}\mu(i_{l+i})$ for all $l\in \{1,\ldots, r\}$ [modulo $r$].

Not all stable mechanisms are locally non-bossy. For instance, the {\it{school-median}} stable mechanism (Klaus and Klijn, 2006) does not satisfy this property (see \autoref{median}).\\

It is well-known that ${\rm{DA}}$ is the only mechanism defined in ${\mathcal{P}}$ that is both stable and strategy-proof (Alcalde and Barber\`a, 1994). Since local group strategy-proofness is stronger than strategy-proofness, and local non-bossiness and strategy-proofness ensure local
group strategy-proofness (\autoref{papai}), the following result is a direct consequence of \autoref{loc_nonb}.\medskip

\begin{corollary}\label{cor1}
${\rm{DA}}$ is the only mechanism that is stable and locally group strategy-proof.\\
\end{corollary}

\section{A characterization of DA via local non-bossiness}\label{axioma}

In this section, we provide a characterization of DA using axioms that do not involve stability. Our result extends the characterization of DA given by Ehlers and Klaus (2016, Theorem 1) from one-to-one matching problems to many-to-one contexts. 

To maintain Ehlers and Klaus' (2016) context, we analyze a variable-population school choice model. Hence, given a population $\overline{N}$ of students and a set $S$ of schools with capacities $q=(q_s)_{s\in S}$, we consider mechanisms that depend only on the set $N\subseteq\overline{N}$ of students in the market and their preferences $(P_i)_{i\in N}$. To simplify notation, we consider ${\mathcal{P}}={\mathcal{L}}^{\overline{N}}$ as the preference domain, although only the preferences of the agents in the market will be relevant at the moment of determining a matching. Let ${\mathcal{N}}$ be the collection of non-empty subsets of $\overline{N}$. 

Given $N\in {\mathcal{N}}$, let $\mathcal{M}(N)$ be the set of matchings compatible with $[N,S,q]$. That is, $\mathcal{M}(N)$ is the family of functions $\mu:N\rightarrow S\cup\{s_0\}$ such that $\vert \{i\in N: \mu(i)=s\}\vert \leq q_s$ for all $s\in S$. Let $\mathcal{M}^*=\bigcup_{N\in {\mathcal{N}}} \mathcal{M}(N)$. For each $i\in \overline{N}$ and $P_i\in {\mathcal{L}}$, let $$A(P_i)=\{s\in S: sP_i s_0\}$$ be the set of schools that are admissible under $P_i$. Moreover, denote by $P^t_i$ the strict linear order defined on $A(P_i)$ such that, for every $s,s'\in A(P_i)$, it holds that $s P^t s'$ if and only if $s P_i s'$. In other words, $P^t_i$ is the truncation of $P_i$ at $s_0$. 

A {\bf{mechanism}} is a function $\Phi:{\mathcal{N}}\times \mathcal{P}\rightarrow \mathcal{M}^*$ that associates to each $(N, P)\in {\mathcal{N}}\times \mathcal{P}$ a matching in ${\mathcal{M}}(N)$ that depends only on $[N,S,q,(P^t_i)_{i\in N}]$. In other words, the matching $\Phi(N,P)$ can only be affected by the set of schools and their capacities, by the set of students in the market, and by the way in which these students rank their admissible schools. This last assumption allows us to avoid making $\Phi(N,P)$ depend on the alternatives that students in $N$ consider inadmissible, or on the preferences of individuals in $\overline{N}\setminus N$.\medskip
 
\newpage

Consider the following properties:
\begin{itemize}
\item $\Phi$ is {\textbf{weakly non-wasteful}} when for all $N\in {\mathcal{N}}$, $i\in N$, $P\in{\mathcal{P}}$, and $s\in S$, if $s P_i \Phi_i(N,P)$ and $\Phi_i(N,P)=s_0$, then $\vert\Phi_s(N,P)\vert=q_s$.
\item $\Phi$ is {\textbf{population-monotonic}} if for all $N,N' \in {\mathcal{N}}$ such that $N\subseteq N'$, $i\in N$, and $P\in   {\mathcal{P}}$  we have that $\Phi_i(N,P) R_i \Phi_i(N',P)$.
\end{itemize}

Under weak non-wastefulness, no unassigned student prefers a school that did not fill its places. Population-monotonicity guarantees that when the set of students enlarges, those who were initially present are (weakly) worse off.\\

By definition, for every mechanism $\Phi:{\mathcal{N}}\times{\mathcal{P}}\rightarrow {\mathcal{M}}^*$ and $(N,P)\in {\mathcal{N}}\times{\mathcal{P}}$, the matching $\Phi(N,P)$ depends only on $[N,S,q, (P^t_i)_{i\in N}]$. Hence, given a preference profile $P\in {\mathcal{P}}$ such that $A(P_i)=\{s\}$ for all $i\in {\overline{N}}$, the mapping that associates to each $N\subseteq \overline{N}$ the set $\Phi_s(N,P)$ can be viewed as a {\it{choice function}} for school $s$, in the sense that it determines the students assigned to $s$ for every subset of candidates in $\overline{N}$. \medskip

\begin{definition}
A mechanism  $\Phi$ satisfies the {\textbf{S-weaker axiom of revealed preference}} ({\textbf{S-WrARP}}) when for every $i,j\in \overline{N}$, $s\in S$, and $N,N'\in {\mathcal{N}}$ such that $i,j\in N\cap N'$ and
$ \vert N\vert=\vert N' \vert=q_s+1$, the following property holds when $P\in {\mathcal{P}}$ satisfies $A(P_k)=\{s\}$ for all $k\in \overline{N}$, 
$$\quad \Big[\, i\in \Phi_s(N,P),\,\,\quad j\in \Phi_s(N',P)\setminus \Phi_s(N,P)\,\Big]\,\Longrightarrow \, \big[\, i \in \Phi_s(N',P) \,\big].\quad $$
\end{definition}

$\,$

S-WrARP is a restricted version of the {\it{weaker axiom of revealed preference}} introduced by Jamison and Lau (1973), but applied to schools' choice functions induced by a mechanism. More precisely, assuming that some school $s$ is the only admissible alternative, consider a situation in which two students, $i$ and $j$, can be assigned to $s$ from two different sets of size $q_s+1$, namely $N$ and $N'$. The property S-WrARP ensures that, when the mechanism $\Phi$ is implemented, if $i$, but not $j$, is assigned to $s$ when the set $N$ is considered, then whenever $j$ is assigned to $s$ when $N'$ is considered, $i$ is also assigned to $s$. This axiom will allow us to appeal to the results of Chambers and Yenmez (2018) to construct schools' priority orders that are consistent with the assignments determined by $\Phi$. \\

\begin{definition}
A mechanism  $\Phi$ is {\textbf{weakly locally non-bossy}} if for all $N\in {\mathcal{N}}$, $i\in N$, $P\in{\mathcal{P}}$, $P'_i\in {\mathcal{L}}$, and $s\in S$, we have that $\Phi_i(N,P) =  \Phi_i(N,(P'_i, P_{-i}))=s$ implies that $\Phi_s(N,P) =  \Phi_s(N,(P'_i, P_{-i}))$.\medskip
\end{definition}

Weak local non-bossiness restricts local non-bossiness to the set of (real) schools $S$. \\

Given a priority profile $\succ$, let ${\rm{DA}}^\succ$ be the mechanism that associates to each $(N,P)\in{\mathcal{N}}\times \mathcal{P}$ the student-optimal stable matching of $[N,S,\succ^N,q,(P_i)_{i\in N}]$, where $\succ^N$ is the priority profile induced by $\succ$.\footnote{That is, if $\succ=(\succ_s)_{s\in S}$,  the induced priority profile $\succ^N=(\succ^N_s)_{s\in S}$ is such that, for all $i,j\in N$ and $s\in S$, $i \succ^N_s j$ whenever $i \succ_s j$.} 
When each school has only one available seat, Ehlers and Klaus (2016) demonstrate that a mechanism $\Phi$ satisfies individual rationality, weak non-wastefulness, population-monotonicity, and strategy-proofness if and only if $\Phi={\rm{DA}}^{\succ}$ for some priority profile $\succ$. Furthermore, they show that this characterization does not hold in the general case in which schools have non-unit capacities.

The following result shows that the four axioms of Ehlers and Klaus (2016) jointly with S-WrARP and weak local non-bossiness fully characterize the family of mechanisms ${\rm{DA}}^\succ$, where $\succ$ is an arbitrary priority profile. Note that the two new axioms are trivially satisfied when each school has only one available seat. \medskip

\begin{theorem}\label{coro-bando}
A mechanism $\Phi:{\mathcal{N}}\times \mathcal{P}\rightarrow \mathcal{M}^*$ is individually rational, weakly non-wasteful,  population-monotonic, strategy-proof, weakly locally non-bossy, and satisfies S-WrARP if and only if $\Phi={\rm{DA}}^{\succ}$ for some priority profile $\succ$.\medskip
\end{theorem}

\noindent The proof of this result and examples that show that the six axioms are independent are provided in \autoref{appT2}.\\


In our characterization of DA, weak local non-bossiness is crucial to manage the schools with more than one seat available. Indeed, Example 1 of Ehlers and Klaus (2016)---which we reproduce below---allows us to describe a mechanism that satisfies all the axioms of \autoref{coro-bando} but weak local non-bossiness. \\

\begin{example}\label{EK-1} Let $\overline{N}=\{1,2,3,4\}$, $S=\{s_1,s_2\}$, and $(q_{s_1},q_{s_2})=(2,1)$. Consider the priority orders $\succ=(\succ_{s_1},\succ_{s_2})$ and $\succ'=(\succ'_{s_1},\succ_{s_2})$ characterized by  $\succ_{s_1}:1, 2, 3, 4$, $\succ_{s_2}: 1, 2, 3,4$, and $\succ'_{s_1}: 1, 2, 4, 3$. Denote by $top(P_i)$ the most preferred alternative of a student $i$ when her preferences are $P_i\in {\mathcal{L}}$. Let $\Omega:{\mathcal{N}}\times {\mathcal{P}}\rightarrow {\mathcal{M}}^*$ be the mechanism such that: $$\Omega(N,P)=
\left\{
\begin{array}{ll}
{\rm{DA}}^{\succ'}(N,P), &\quad\quad \mbox{when}\, \, \{1,2\}\subseteq N\,\mbox{and}\,top(P_1)=top(P_2)=s_2;\\
{\rm{DA}}^{\succ}(N,P), &\quad\quad \mbox{otherwise}. 
\end{array}
\right.
$$

Ehlers and Klaus (2016, Example 1) show that $\Omega$ is individually rational, weakly non-wasteful,  population-monotonic, and strategy-proof. Furthermore, $\Omega$ satisfies S-WrARP, because for each profile of preferences schools choose students from a set following a priority order. However, $\Omega$ is weakly locally bossy. Indeed, when preferences are given by 
$$P_1: s_2, s_0, s_1, \quad\quad  P_2: s_1, s_0, s_2, \quad\quad P_3: s_1, s_0, s_2 \quad\quad P_4: s_1, s_0, s_2,$$
we have that $\Omega(\overline{N},P)={\rm{DA}}^{\succ}(\overline{N},P)= ((1,s_2), (2,s_1), (3,s_1), (4,s_0)).$ If $P'_2: s_2,s_1,s_0$,  $\Omega(\overline{N},(P'_2,P_{-2}))={\rm{DA}}^{\succ'}(\overline{N},(P'_2,P_{-2}))= ((1,s_2), (2,s_1), (3,s_0), (4,s_1)).$ Therefore, the student $2$ can change her colleagues without modifying her school.\footnote{This argument also shows that $\Omega$ is not group strategy-proof. Indeed, the coalition $\{2,4\}$ can report preferences $(P'_2,P_4)$ to improve the assignment of $4$ without hurting $2$. Hence, the same example allows us to show that group strategy-proofness is a critical axiom in \autoref{coro-bando2}.} \hfill $\Box$\\
\end{example}

A priority profile $\succ$ has an {\bf{Ergin-cycle}} when for some $s,s'\in S$ and $i,j,k\in \overline{N}$ we have that $i\succ_{s} j \succ_s k \succ_{s'} i$ and there are disjoint sets $N_{s}, N_{s'}\subseteq \overline{N}\setminus\{i,j,k\}$, with $\vert N_{s}\vert =q_{s}-1$ and $\vert N_{s'}\vert =q_{s'}-1$, such that  $j$ has lower priority in $s$ than everyone in $N_{s}$ and $i$ has lower priority in $s'$ than everyone in $N_{s'}$. The priority profile $\succ$ is {\bf{acyclic}} when it does not have Ergin-cycles. \\

\begin{corollary}\label{coro-bando2}
A mechanism $\Phi:{\mathcal{N}}\times \mathcal{P}\rightarrow \mathcal{M}^*$ is individually rational, weakly non-wasteful,  population-monotonic, group strategy-proof, and satisfies S-WrARP if and only if $\Phi={\rm{DA}}^{\succ}$ for some acyclic priority profile $\succ$.\medskip
\end{corollary}

\begin{proof}
Group strategy-proofness ensures that strategy-proofness and non-bossiness hold (Pap\'ai, 2000). Hence, \autoref{coro-bando} implies that a mechanism $\Phi$ satisfying the five axioms described above coincides with ${\rm{DA}}^{\succ}$, for some priority profile $\succ$. In particular, as ${\rm{DA}}^{\succ}$ is group strategy-proof, $\succ$ is acyclic (Ergin, 2002).  Reciprocally, given an acyclic priority profile $\succ$, it follows from Ergin (2002, Theorem 1) that ${\rm{DA}}^\succ$ is group strategy-proof (and \autoref{coro-bando} ensures that ${\rm{DA}}^\succ$ satisfies the other axioms).
\end{proof}

\medskip

\section{Mechanisms under preferences over colleagues}\label{ext}

Most of the literature on school choice assumes that students only care about the school to which they are assigned. However, it is  unrealistic that students do not care about the assignment of others. In particular, it has been empirically documented that the identity of colleagues is an important consideration when students are choosing a school (Rothstein, 2006; Abdulkadiro\u{g}lu et al., 2020; Allende, 2021; Che et al., 2022; Beuermann et al., 2023; Cox et al., 2023).

When students have arbitrary preferences over the set of matchings, many of the results in the literature break down. Specifically, a stable matching may not exist, {\it{even when students only care about the identity of the other students assigned to her school}} (Echenique and Yenmez, 2007). One way to overcome this problem is to restrict the domain of preferences. For instance, suppose that students have \textit{school-lexicographic preferences}, meaning that they are primarily concerned with the school to which they are assigned (Sasaki and Toda, 1996; Dutta and Mass\'o, 1997). 
In this domain, some of the standard results are recovered, such as the existence of a stable matching. However, a stable and strategy-proof mechanism may not exist (Duque and Torres-Mart\'{\i}nez, 2023). One may wonder if there is a subdomain of school-lexicographic preferences where stability and strategy-proofness are compatible. In this section, we show that within the subdomain of school-lexicographic preferences where only the schoolmates matter, the local non-bossiness of ${\rm{DA}}$ implies that this rule induces a stable and strategy-proof mechanism. \medskip

We first define the domain of school-lexicographic preferences, and then the subdomain where there exists a stable and strategy-proof mechanism.\medskip

\begin{definition}\label{lexi} Given a school choice context $[N, S, \succ, q]$, a {\textbf{school-lexicographic preference}} for student $i\in N$ is a complete and transitive preference relation $\unrhd_i$ defined on ${\mathcal{M}}$ such that, for every pair of matchings $\mu,\eta\in {\mathcal{M}}$, the following properties hold:
\begin{itemize}
\item If $\mu(i)\neq \eta(i)$, then either $\mu \rhd_i  \eta$ or $\eta \rhd_i \mu$, where $\rhd_i$ is the strict part of $\unrhd_i$.
\item If $\mu \rhd_i \eta $ and $\mu(i)\neq\eta(i)$, then $\mu' \rhd_i \eta'$ for all matchings $\mu',\eta'\in {\mathcal{M}}$ such that $\mu'(i)=\mu(i)$ and $\eta'(i)=\eta(i)$.
\end{itemize} 
Let $\mathcal{D}$ be the set of preference profiles $\unrhd=(\unrhd_i)_{i\in N}$ satisfying these properties. 
\end{definition}

\medskip

The first condition states that a student cannot be indifferent between two matchings where she is assigned to different schools. By the second condition, if a student prefers a matching $\mu$ over $\eta$, and she is assigned to different schools in these matchings, then she should prefer every matching where she is assigned to the same school as in $\mu$ to every matching where she is assigned to the same school as in $\eta$.  

Notice that, in the preference domain ${\mathcal{D}}$ no restrictions are imposed on the ranking of two matchings in which a student is assigned to the same school. Let $P(\unrhd)=(P_i(\unrhd_i))_{i\in N}\in {\mathcal{P}}$ be the preferences over schools induced by $\unrhd=(\unrhd_i)_{i\in N} \in {\mathcal{D}}$ through the following rule: given $s,s'\in S\cup\{s_0\}$, $s \,P_i(\unrhd_i) \, s'$ if and only if there exist $\mu,\mu'\in {\mathcal{M}}$ such that $\mu(i)=s$, $\mu'(i)=s'$, and $\mu \rhd_i \mu'$. The second property in \autoref{lexi} ensures that $P(\unrhd)$ is well-defined.\
 
One can naturally extend the concepts of stability and strategy-proofness to the domain of school-lexicographic preferences using the induced preferences over schools. Given ${\mathcal{D}}'\subseteq {\mathcal{D}}$, a {\textbf{stable mechanism}} associates to each preference profile $\unrhd\in {\mathcal{D}}'$ a stable matching of $[N, S, \succ, q, \unrhd]$.  A mechanism $\Gamma$ is {\textbf{strategy-proof}} in ${\mathcal{D}}'$ when there is no student $i$ such that $\Gamma_i(\unrhd'_i, \unrhd_{-i}) \rhd_i  \Gamma_i(\unrhd)$ for some $\unrhd$ and $(\unrhd'_i, \unrhd_{-i})$ in ${\mathcal{D}}'$. \\

School-lexicographic preferences have no effect on stability, because the school choice problems $[N, S, \succ, q, \unrhd]$ and $[N, S, \succ, q, P(\unrhd)]$ have the same set of stable matchings  (Sasaki and Toda, 1996; Dutta and Mass\'o, 1997; Fonseca-Mairena and Triossi, 2023). However, there are strong effects on incentives as there may not exist a stable and strategy-proof mechanism (Duque and Torres-Mart\'{\i}nez, 2023). 

Within the domain of school-lexicographic preferences, let us consider a scenario where students are solely concerned about the composition of the school to which they are assigned. Consequently, a student will be indifferent between two matchings when she is assigned to the same school with the same colleagues. This is the idea of the following subdomain of preferences.\medskip

\begin{definition} Given $[N,S,\succ,q]$, the domain of {\textbf{school-lexicographic preferences over colleagues}} $\mathcal{D}_{{\rm{c}}}\subseteq \mathcal{D}$ is the set of profiles $(\unrhd_i)_{i\in N}$ such that $\mu \rhd_i \eta$ and $\mu(i)=  \eta(i)=s$ imply that $\mu(s)\neq \eta(s)$. \medskip
\end{definition}

In words, when $(\unrhd_i)_{i\in N}$ belongs to $ \mathcal{D}_{{\rm{c}}}$, each student $i$ strictly prefers $\mu$ to $\eta$ only when she is assigned to different schools or to the same school with different colleagues.\\

The next result shows that there is a stable and strategy-proof mechanism in $ \mathcal{D}_{{\rm{c}}}$, which is the student-optimal stable mechanism applied to $[N, S, \succ, q, P(\unrhd)]$.\\

\begin{theorem}\label{propo-ex}
In every school choice context $(S, N, \succ, q )$, the mechanism $\overline{\rm{DA}}:\mathcal{D}_{{\rm{c}}}\rightarrow \mathcal{M}$ defined by $\overline{\rm{DA}}(\unrhd)={\rm{DA}}(P(\unrhd))$ is stable and strategy-proof.\medskip
\end{theorem}

\begin{proof} Fix $\unrhd\in \mathcal{D}_{{\rm{c}}}$. The stability of $\overline{\rm{DA}}(\unrhd)$ is a consequence of the fact that ${\rm{DA}}:{\mathcal{P}}\rightarrow {\mathcal{M}}$ is stable, because $[N, S, \succ, q, \unrhd]$ and $[N, S, \succ, q, P(\unrhd)]$ have the same stable matchings. By contradiction, assume that $\overline{\rm{DA}}$ is not strategy-proof in $\unrhd$. Hence, there exists $i\in N$ such that $\overline{\rm{DA}}(\unrhd'_i, \unrhd_{-i})\,\rhd_i\, \overline{\rm{DA}}(\unrhd)$ for some $\unrhd'_i$ such that $\unrhd'\equiv (\unrhd'_i, \unrhd_{-i})$ belongs to $\mathcal{D}_{{\rm{c}}}$. Since $\overline{\rm{DA}}(\unrhd'_i, \unrhd_{-i})\neq \overline{\rm{DA}}(\unrhd)$, we have two cases:
\begin{itemize}
\item {\it{Case 1: $\overline{\rm{DA}}_i(\unrhd'_i, \unrhd_{-i})\neq \overline{\rm{DA}}_i(\unrhd)$.}} The definitions of $\overline{\rm{DA}}$ and $P_i(\unrhd_i)$ ensure that $${\rm{DA}}_i(P_i(\unrhd'_i), P_{-i}(\unrhd))\,\,P_i(\unrhd_i)\,\,{\rm{DA}}_i(P(\unrhd)).$$ This implies that ${\rm{DA}}$ is not strategy-proof in ${\mathcal{P}}$. A contradiction.
\item {\it{Case 2: $\overline{\rm{DA}}_i(\unrhd'_i, \unrhd_{-i})=\overline{\rm{DA}}_i(\unrhd)=s$.}} The definitions of $\overline{\rm{DA}}$ and $\mathcal{D}_{{\rm{c}}}$ ensure that 
\begin{eqnarray*}
{\rm{DA}}_i(P_i(\unrhd'_i), P_{-i}(\unrhd))&=&{\rm{DA}}_i(P(\unrhd))\,=\,s,\\ {\rm{DA}}_s(P_i(\unrhd'_i), P_{-i}(\unrhd))&\neq &{\rm{DA}}_s(P(\unrhd)).
\end{eqnarray*} This implies that ${\rm{DA}}$ is locally bossy, which contradicts \autoref{loc_nonb}.
\end{itemize}

Therefore, the mechanism $\overline{\rm{DA}}:\mathcal{D}_{{\rm{c}}}\rightarrow \mathcal{M}$  is strategy-proof. 
\end{proof}

$\,$

Every stable mechanism $\Phi:{\mathcal{P}}\rightarrow {\mathcal{M}}$ induces a stable mechanism on $\mathcal{D}_{{\rm{c}}}$ by the rule that associates each $\unrhd\in \mathcal{D}_{{\rm{c}}}$ with the matching $\Phi(P(\unrhd))$. Evidently, many other stable mechanisms can be defined in $\mathcal{D}_{{\rm{c}}}$ using the information that students reveal about how they rank their colleagues. Despite this multiplicity, the uniqueness result of Alcalde and Barber\`a (1994) can be extended to the preference domain $\mathcal{D}_{{\rm{c}}}$.\\

\begin{corollary}\label{unique}
On the preference domain $\mathcal{D}_{{\rm{c}}}$ the mechanism $\overline{\rm{DA}}$ is the only one that is stable and strategy-proof.\medskip
\end{corollary}

\begin{proof} It follows from \autoref{propo-ex} that $\overline{\rm{DA}}$ is stable and strategy-proof on the preference domain $\mathcal{D}_{{\rm{c}}}$. By contradiction, assume that there is a stable and strategy-proof mechanism $\Omega:\mathcal{D}_{{\rm{c}}}\rightarrow {\mathcal{M}}$ such that $\Omega(\unrhd)\neq \overline{\rm{DA}}(\unrhd)$ for some $\unrhd\in \mathcal{D}_{{\rm{c}}}$. Since $\mathcal{D}_{{\rm{c}}}$ only includes school-lexicographic preferences, $\overline{\rm{DA}}(\unrhd)$ pairs each student to the most preferred alternative in $S\cup\{s_0\}$ that she can obtain in a stable outcome of $(S, N, \succ, q, \unrhd )$.\footnote{That is, there is no stable matching $\mu$ such that $\mu(i)\,P_i(\unrhd_i)\,\overline{\rm{DA}}_i(\unrhd)$ for some $i\in N$.} Thus, $\overline{\rm{DA}}_i(\unrhd) \rhd_i \Omega_i(\unrhd)$ for some $i\in N$. In particular,  $\overline{\rm{DA}}_i(\unrhd)$ is a school. 

Let $P'_i$ be the preferences defined on $S\cup \{s_0\}$ for which $\overline{\rm{DA}}_i(\unrhd)$  is the only acceptable school (i.e., $s_0 P'_i s$ for all school $s\neq \overline{\rm{DA}}_i(\unrhd)$). Fix $\unrhd'=(\unrhd'_{i},\unrhd_{-i})\in \mathcal{D}_{{\rm{c}}}$ such that $P_i(\unrhd'_i)=P_i'$. Since the problems $(S, N, \succ, q,\unrhd' )$  and $(S, N, \succ, q, (P'_i, P_{-i}(\unrhd)))$ have the same stable matchings, and $\overline{\rm{DA}}(\unrhd)$ is stable under $(P'_i, P_{-i}(\unrhd))$, the definition of $P'_i$ and the Rural Hospital Theorem (Roth, 1984b; Gale and Sotomayor, 1985) imply that $\Omega_i (\unrhd'_{i},\unrhd_{-i})=\overline{\rm{DA}}_i(\unrhd)$. Therefore, $\Omega_i (\unrhd'_{i},\unrhd_{-i}) \rhd_{i} \Omega_i(\unrhd)$, which contradicts the strategy-proofness of $\Omega$.
\end{proof}

\medskip

It is enough for the preferences of just one student to be school-lexicographic but not school-lexicographic over colleagues to prevent the existence of a stable and strategy-proof mechanism. The next example formalizes this claim by following the arguments made in the proof of Theorem 1 of Duque and Torres-Mart\'{\i}nez (2023).\\

\begin{example}\label{EX2}
Let $N=\{1,2,3,4,5\}$, $S=\{s_1,s_2,s_3,s_4\}$, $q_{s_1}=2$, and $q_{s_k}=1$ for all $k\neq 1$. Assume that schools' priorities satisfy the following conditions:
$$\succ_{s_1}: 4,2,1,3,5,  \quad\quad \succ_{s_2}:3,2,\ldots \quad\quad  \succ_{s_3}:1,2,\ldots \quad\quad \succ_{s_4}:2, 5,\ldots.$$ 

Let ${\mathcal{D}}_1$ be the preference domain in which student $1$ has school-lexicographic preferences and the other students have school-lexicographic preferences over colleagues. We claim that no stable mechanism defined in ${\mathcal{D}}_1$ is strategy-proof. 

Consider a preference profile $\unrhd=(\unrhd_i)_{i\in N}\in\mathcal{D}_1$ such that:
$$P_1(\unrhd_1):\, s_3\ldots; \quad\quad\quad P_2(\unrhd_2):\, s_2,s_1,\ldots; \quad\quad\quad  P_3(\unrhd_3):\, s_1,s_2,
\ldots;$$ $$ P_4(\unrhd_4):\,  s_1,\ldots; \quad\quad\quad P_5(\unrhd_5):\, s_4,\ldots.$$

Note that $[N,S, \succ,  q,P(\unrhd)]$ has only two stable matchings:
\begin{eqnarray*}
\mu&=&((1,s_3), (2,s_1), (3,s_2), (4,s_1), (5,s_4)),\\ 
\eta&=&((1,s_3), (2,s_2), (3,s_1), (4,s_1), (5,s_4)).
\end{eqnarray*}

Since there are no restrictions on how student $1$ ranks $\mu$ and $\eta$, because she is assigned to the same school in both matchings and has  school-lexicographic preferences, assume that she strictly prefers $\mu$ to $\eta$ under $\unrhd_1$.

It is not difficult to verify that $\mu$ and $\eta$ are the only stable matchings of $[N,S, \succ,  q,\unrhd]$. Hence, if a mechanism $\Phi:\mathcal{D}_1\rightarrow {\mathcal{M}}$ is stable, then $\Phi(\unrhd)\in \{\mu,\eta\}$.

Suppose that $\Phi(\unrhd)=\mu$. If  $\unrhd'_{2}$ is a school-lexicographic preference over colleagues such that $P_2'\equiv P_2(\unrhd'_2)$ satisfies $P'_2:s_2,s_4,\ldots$, then $\eta$ is the only stable matching of $[N,S, \succ,  q,(\unrhd'_2,\unrhd_{-2})]$. Hence, the student $2$ has incentives to manipulate $\Phi$, because $s_2 P_2(\unrhd_2)s_1$ implies that $\Phi(\unrhd'_2,\unrhd_{-2})=\eta \rhd_2 \mu= \Phi(\unrhd)$.

Suppose that $\Phi(\unrhd)=\eta$. If $\unrhd'_{1}$ is such that $P_1'\equiv P_1(\unrhd'_1)$ satisfies $P'_1: s_1,s_3,\ldots$, then $\mu$ is the only stable matching of $[N,S, \succ,  q,(\unrhd'_1,\unrhd_{-1})]$. Hence, student $1$ has incentives to manipulate $\Phi$, because $\Phi(\unrhd'_1,\unrhd_{-1})=\mu \rhd_1 \eta= \Phi(\unrhd).$

We conclude that no stable mechanism defined in ${\mathcal{D}}_1$ is strategy-proof.\hfill $\Box$\\
\end{example}

Given a school choice context $[N,S,\succ,q]$ and a domain ${\mathcal{D}}'$ such that ${\mathcal{D}}_{\rm{c}}\subsetneq {\mathcal{D}}'\subseteq {\mathcal{D}}$, to ensure that no mechanism defined in ${\mathcal{D}}'$ is stable and strategy-proof it is crucial that $\succ$ has a cycle in the sense of Ergin (2002).\footnote{See the previous section for the definition of Ergin-cycle. The priority profile of \autoref{EX2} has an Ergin-cycle characterized by $(s,s')=(s_1,s_2)$, $(i,j,k)=(2,1,3)$, $N_{s}=\{4\}$, and $N_{s'}=\emptyset$.} Indeed, the mechanism $\overline{\rm{DA}}$ is stable and strategy-proof in ${\mathcal{D}}$ whenever $\succ$ is Ergin-acyclic (see Duque and Torres-Mart\'{\i}nez, 2023).\\

\section{Concluding remarks}\label{conclu}

In this paper, we have shown that the bossiness of the student-optimal stable mechanism is limited. When ${\rm{DA}}$ is implemented, it is well known that a student can affect another student's assignment by changing her preferences without modifying her school. We have demonstrated that this holds only for students who are not assigned to the same school. In other words, under ${\rm{DA}}$, a student cannot change her schoolmates without changing her own school. Additionally, for every strategy-proof mechanism, local non-bossiness guarantees that no coalition of students assigned to the same school can misrepresent their preferences to either improve their assignments or maintain their school while modifying their colleagues. Since ${\rm{DA}}$ is not group strategy-proof, our result implies that every successful manipulating coalition of DA includes students from different schools. 

Local non-bossiness is not only interesting in itself but also because of its application to school choice problems with externalities. In this framework, even when students prioritize their assigned school first and then consider the assignment of others, a stable and strategy-proof mechanism may not exist (Duque and Torres-Mart\'{\i}nez, 2023). However, if we restrict preferences to be such that only the own school and the schoolmates matter, it turns out that ${\rm{DA}}$ induces the only stable and strategy-proof mechanism. The crucial property behind this result is precisely its local non-bossiness.\\

\newpage

\appendix

\section{local non-bossiness and its implications}\label{app2}\medskip

We analyze the relationship between local non-bossiness and other incentive properties, such as strategy-proofness, (local) group strategy-proofness, weak non-bossiness (Bando and Imamura, 2016), and a local notion of group non-bossiness (Afacan, 2012).  The following diagram summarizes the causal relationships between these incentive properties that we formalize in this appendix.\\

\begin{figure}[H]
    \centering
\quad{\resizebox{15cm}{!}{
$$\begin{tikzpicture}[>=stealth,node distance=2.5cm and 0.3cm]
    \node [shape=ellipse,fill={rgb:black,1;white,15},draw,rounded corners](c1) {{{Local non-bossiness \,  + \, Strategy-proofness }}};
        \node [shape=ellipse,fill={rgb:black,1;white,15},draw,rounded corners](c0) [above =1.5cm and -2cm of c1]{{{\, Group strategy-proofness \,}}};
    \node [shape=ellipse,fill={rgb:black,1;white,15},draw,rounded corners](c2) [below left =2cm and -2cm of c1]{{{Local group strategy-proofness}}};
    \node [shape=ellipse,fill={rgb:black,1;white,15},draw,rounded corners](c3) [below right =2cm and -2cm of c1]{{{Local group non-bossiness}}};
        \node [](c12) [below right =5cm and -12.7cm of c1]{{{}}};
                \node [](c13) [below right =2.5cm and 4.2cm of c1]{{{}}};
    \node [shape=ellipse,fill={rgb:black,1;white,15},draw,rounded corners](c8) [below right =10cm and -5.6cm of c1]{{{Strategy-proofness}}};
     \node [shape=ellipse,fill={rgb:black,1;white,15},draw,rounded corners](c10) [below left =7.5 cm and -5.6cm of c1]{{{Weak non-bossiness}}};
    \node [shape=ellipse,fill={rgb:black,1;white,15},draw,rounded corners](c9) [below left =5 cm and -5.6cm of c1]{{{Local non-bossiness}}};
     \draw[-{Stealth[length=3mm, width=2mm]}, degil] (c1.north east) to[out=120,in=-30]  node[midway,below = -0.2cm] {{\scriptsize{DA}}\,\,\quad\quad\quad\quad} (c0.south east);
      \draw[-{Stealth[length=3mm, width=2mm]}] (c0.south west) to[out=-150,in=60]  (c1.north west);
    \draw[-{Stealth[length=3mm, width=2mm]}] (c1.south) to[out=-130,in=45]  node[midway,above = 0.2cm] {{\scriptsize{Lemma 1}}\quad} (c2.north);
        \draw[-{Stealth[length=3mm, width=2mm]}] (c1.south) to[out=-50,in=135] node[midway,above = 0.2cm] {\,\,\,{\scriptsize{Lemma 2}}} (c3.north);
          \draw[-{Stealth[length=3mm, width=2mm]}, degil] (c2.south) to[out=-90,in=175] node[midway,above =0cm] {{\scriptsize{\quad \quad \quad\quad\quad\,\,\,Example 4}}} (c9.west);
            \draw[-{Stealth[length=3mm, width=2mm]}, degil] (c3.south east) to[out=-20,in=0] node[midway,below =0.4cm] {{\scriptsize{\quad \quad\quad \quad \quad\quad  \quad\quad \quad   Boston mechanism}}}
            (c8.east);
               \draw[-{Stealth[length=3mm, width=2mm]}, degil] (c8.west) to[out=175,in=-90] node[midway,above =0.4cm] {\quad {\quad\quad\scriptsize{Example 6}}} (c2.south west);
            \draw[-{Stealth[length=3mm, width=2mm]}, degil] (c9.north west) to[out=45,in=185] node[midway,below =0.1cm] {\quad\quad\quad\quad\,\,\,{\scriptsize{Example 5}}} (c3.west);
            \draw[-{Stealth[length=3mm, width=2mm]}, degil] (c10.north) to[out=160,in=-70] node[midway,right=0.2cm] {{\scriptsize{Lemma 3}}\quad\quad\quad} (c9.south west);
             \draw[-{Stealth[length=3mm, width=2mm]}, degil] (c10.south east) to[out=-120,in=30] node[midway,below=-0.2cm] {\quad\quad\quad\,\,\,\,\,\,\,\,{\scriptsize{Lemma 3}}} (c8.north);
            \draw[-{Stealth[length=3mm, width=2mm]}] (c2.south) to[out=-120,in=150] (c8.north west);
                \draw[-{Stealth[length=3mm, width=2mm]}] (c3.south east) to[out=-100,in=0] (c9.east);
                 \draw[-{[length=3mm, width=2mm]}] (c0.south west) to[out=185,in=110] (c12.west);
                \draw[-{Stealth[length=3mm, width=2mm]}] (c12.west) to[out=-70,in=200] (c10.south west);
                 \draw[-{[length=3mm, width=2mm]}] (c1.east) to[out=0,in=90] (c13.south);
                  \draw[-{Stealth[length=3mm, width=2mm]},degil] (c13.south) to[out=-90,in=0] 
                  node[midway,above =0cm] {{\scriptsize{Lemma 3 \quad\quad\quad\quad\quad \quad\quad }}}
                  (c10.east);
\end{tikzpicture}
$$
}}
 \caption{On local non-bossiness and other incentive properties.}
        \label{diag}
\end{figure}

$\,$

With the information provided in \autoref{diag}, it can be easily inferred that all absent causal relationships are not satisfied.\\

\begin{lemma}\label{papai}
If $\Phi:{\mathcal{P}}\rightarrow {\mathcal{M}}$ is a locally non-bossy and strategy-proof mechanism, then it is locally group strategy-proof.\medskip
\end{lemma}

\begin{proof} 
Given  $P\in {\mathcal{P}}$ and $s\in S\cup \{s_0\}$, suppose that there is a coalition $C=\{i_1,\ldots, i_r\}$ contained in $\Phi_s(P)$ such that $\Phi_i(P'_{C}, P_{-C}) R_i  \,\Phi_i(P)$ for all $i\in C$ and for some $P'_C=(P'_j)_{j\in C}$. We will prove that $\Phi_i(P'_{C}, P_{-C}) =  \,\Phi_i(P)$ for all $i\in C$. For each $i\in \Phi_s(P)$, let $\widehat{P}_i$ be the preference relation that places $\Phi_i(P'_{C}, P_{-C})$ at the top and keeps the other schools in the order induced by $P_i$. 

Suppose that $\Phi_{i_1}(P) \neq \Phi_{i_1}({\widehat{P}}_{i_1}, P_{-i_1})$. The strategy-proofness of $\Phi$ implies that $\Phi_{i_1}(P) P_{i_1} \Phi_{i_1}({\widehat{P}}_{i_1}, P_{-i_1})$ and $\Phi_{i_1}({\widehat{P}}_{i_1}, P_{-i_1})\widehat P_{i_1} \Phi_{i_1}(P).$ If $\Phi_{i_1}({\widehat{P}}_{i_1}, P_{-i_1})\neq \Phi_{i_1}(P'_{C}, P_{-C})$, the definition of $\widehat{P}_{i_1}$ implies that $\Phi_{i_1}({\widehat{P}}_{i_1}, P_{-i_1}) P_i \Phi_{i_1}(P)$. Hence, $\Phi_{i_1}(P) P_{i_1} \Phi_{i_1}({\widehat{P}}_{i_1}, P_{-i_1}) P_{i_1} \Phi_{i_1}(P)$, which is not possible. As a consequence, 
$\Phi_{i_1}(P) \neq \Phi_{i_1}({\widehat{P}}_{i_1}, P_{-i_1})$ ensures that $\Phi_{i_1}({\widehat{P}}_{i_1}, P_{-i_1})= \Phi_{i_1}(P'_{C}, P_{-C})$, which in turn implies that $\Phi_{i_1}(P) P_{i_1} \Phi_{i_1}(P'_{C}, P_{-C})$, a contradiction. 

Since $\Phi_{i_1}(P) =\Phi_{i_1}({\widehat{P}}_{i_1}, P_{-i_1})$, the  local non-bossiness of $\Phi$ implies that $  \Phi_{i}(P)= \Phi_{i}(\widehat{P}_{i_1}, P_{-i_1})$ for all $i\in \Phi_s(P)$. Repeating the arguments above for students $i_2$,\ldots, $i_r$ we obtain that $ \Phi_{i}(P)=\Phi_{i}(\widehat{P}_{C}, P_{-C})$ for all $i\in \Phi_s(P).$ 

On the other hand, as $\Phi_{i_1}(P'_{C}, P_{-C})$ is the most preferred alternative under $\widehat{P}_{i_1}$, it follows that $\Phi_{i_1}(P'_{C}, P_{-C})\widehat R_{i_1}\Phi_{i_1}(\widehat{P}_{i_1}, P'_{C\setminus\{i_1\}}, P_{-C})$. If $\Phi_{i_1}(P'_{C}, P_{-C})\neq \Phi_{i_1}(\widehat{P}_{i_1}, P'_{C\setminus\{i_1\}}, P_{-C})$, then $\Phi_{i_1}(P'_{C}, P_{-C})\widehat P_{i_1}\Phi_{i_1}(\widehat{P}_{i_1}, P'_{C\setminus\{i_1\}}, P_{-C})$, which contradicts the strategy-proofness of $\Phi$. Moreover, as $\Phi$ is locally non-bossy, $\Phi_{i}(P'_{C}, P_{-C})=\Phi_{i}(\widehat{P}_{i_1}, P'_{C\setminus\{i_1\}}, P_{-C})$ for all $i\in \Phi_s(P)$. Repeating this argument for students $i_2$,\ldots, $i_r$, we obtain that $\Phi_{i}(P'_{C}, P_{-C})=\Phi_i(\widehat{P}_{C}, P_{-C})$ for all $i\in \Phi_s(P).$ 

Therefore, $\Phi_i(P'_{C}, P_{-C}) = \Phi_i(P)$ for every student $i\in C$.
\end{proof}

\medskip

In \autoref{ex2_app} we show that the converse of \autoref{papai} does not hold.\\

Afacan (2012) introduces the following extension of non-bossiness to groups of students:  a mechanism $\Phi$ is {\bf{group non-bossy}} when for every $P\in {\mathcal{P}}$ and $C\subseteq N$, if there exists $P'_{C}\in {\mathcal{L}}^{\vert C\vert}$ such that $\Phi_i(P) =  \Phi_i(P'_C, P_{-C})$ for all $i\in C$, then $\Phi(P) =  \Phi(P'_C, P_{-C})$. 

We will consider a local version of this property: a mechanism $\Phi$ is {\bf{locally group non-bossy}} when for every $s\in S\cup \{s_0\}$, $P\in {\mathcal{P}}$, and $C\subseteq \Phi_s(P)$, if there exists $P'_{C}\in {\mathcal{L}}^{\vert C\vert}$ such that $\Phi_i(P) =  \Phi_i(P'_C, P_{-C})$ for all $i\in C$, then $\Phi_s(P) =  \Phi_s(P'_C, P_{-C})$. Hence, local group non-bossiness ensures that no coalition of students assigned to the same school can change the assignment of some of their colleagues by reporting different preferences and without changing their school.\medskip

\begin{lemma}\label{afacan}
If $\Phi:{\mathcal{P}}\rightarrow {\mathcal{M}}$ is a locally non-bossy and strategy-proof mechanism, then it is locally group non-bossy.\medskip
\end{lemma}

\begin{proof} 
Given $s\in S\cup \{s_0\}$, $P\in {\mathcal{P}}$, $C\subseteq \Phi_s(P)$, and $P'_{C}\in {\mathcal{L}}^{\vert C\vert}$, suppose that $\Phi_i(P) =  \Phi_i(P'_C, P_{-C})$ for all $i\in C$. Since $\Phi_i(P) =  \Phi_i(P'_C, P_{-C})$ ensures that $ \Phi_i(P'_C, P_{-C}) R_i \Phi_i(P)$, the arguments made in the proof of \autoref{papai} guarantee that $\Phi_j(P) =  \Phi_j(P'_C, P_{-C})$ for all $j\in \Phi_s(P)$. Hence, $\Phi_s(P)\subseteq \Phi_s(P'_C, P_{-C})$. Moreover, swapping the roles of $P$ and $(P'_C, P_{-C})$, we obtain $ \Phi_s(P'_C, P_{-C})\subseteq \Phi_s(P)$. Therefore, $\Phi_s(P)=\Phi_s(P'_C, P_{-C})$.
\end{proof}

\medskip

The converse of \autoref{afacan} does not hold. Indeed, consider the {\it{Boston mechanism}}, also known as the Immediate Acceptance mechanism. This mechanism runs similarly to DA, with the difference that at each step accepted students are definitively matched to schools. Although it is not strategy-proof (Abdulkadiro\u{g}lu and S\"{o}nmez, 2003), it is  locally group non-bossy. 
Formally, denoting by ${\mathcal{B}}:{\mathcal{P}}\rightarrow {\mathcal{M}}$ this mechanism, suppose that ${\mathcal{B}}_i(P) =  {\mathcal{B}}_i(P'_C, P_{-C})=s$ for all students in a coalition  $C\subseteq {\mathcal{B}}_s(P)$. Let $\widehat{P}_{C}=(\widehat{P}_i)_{i\in C}$ be such that $s$ is the most preferred alternative under each $\widehat{P}_i$. It is not difficult to verify that  ${\mathcal{B}}_s(P) =  {\mathcal{B}}_s(\widehat{P}_C, P_{-C})={\mathcal{B}}_s(P'_C, P_{-C})$. \medskip

In what follows, we present examples that allow us to relate local non-bossiness and other properties.\\ 

\begin{example}\label{ex1_app}{\it{A stable mechanism that is neither locally non-bossy nor strategy-proof}}.

Let $N=\{1,2,3\}$, $S=\{s_1,s_2\}$, $(q_{s_1},q_{s_2})=(2,1),$ $\succ_{s_1}:1,2,3$, and $\succ_{s_2}:3,2,1$.  

Let $\overline{P}=(\overline{P}_1,\overline{P}_2,\overline{P}_3)$ be such that:
\begin{flalign*}
\overline{P}_1: s_1, s_2, s_0, \quad\quad\quad
 \overline{P}_2: s_2, s_1, s_0, \quad\quad\quad
 \overline{P}_3: s_1, s_2, s_0.
\end{flalign*}

Since the {\it{school-optimal}} stable matching of $[N, S, \succ, q, \overline P]$ is $${\rm{DA}}^S(\overline{P})=((1,s_1),(2,s_1),(3,s_2)),$$ it differs from ${\rm{DA}}(\overline{P})=((1,s_1),(2,s_2),(3,s_1)).$ Let $\Omega:{\mathcal{P}}\rightarrow {\mathcal{M}}$ be the stable mechanism such that $\Omega(P)={\rm{DA}}(P)$ when $P\neq \overline{P}$, and $\Omega(\overline{P})={\rm{DA}}^S(\overline{P})$. 

We claim that $\Omega$ is locally bossy. Let $P_1$ be such that ${P}_1: s_2, s_1, s_0$. It is easy to see that $\Omega(P_1,\overline{P}_{-1})={\rm{DA}}(P_1,\overline{P}_{-1})=((1,s_1),(2,s_2),(3,s_1))$. Since $\Omega(\overline{P}_1,\overline{P}_{-1})={\rm{DA}}^S(\overline{P})$, when student $1$'s preferences change from $P_1$ to $\overline P_1$, she remains assigned to $s_1$, but her schoolmates change as $\Omega_{s_1}(\overline{P}_1, \overline{P}_{-1})=\{1,2\}\neq \{1,3\}= \Omega_{s_1}(P_1,\overline{P}_{-1}).$ 

Since DA is the only stable and strategy-proof mechanism (Alcalde and Barber\`a, 1994), it follows that $\Omega$ does not satisfy strategy-proofness. \hfill $\Box$  \\
\end{example}

Since Gale's Top Trading Cycles mechanism (Shapley and Scarf, 1974) is non-bossy and unstable, \autoref{ex1_app} implies that local non-bossiness and stability are independent properties. \\

Bando and Imamura (2016) introduce the following property: a mechanism $\Phi$ is {\bf{weakly non-bossy}} when for every $P\in {\mathcal{P}}$, $i\in N$, and $P'_i\in {\mathcal{L}}$, $\Phi_i(P) =  \Phi_i(P'_i, P_{-i})$ implies that $\Phi_{s_0}(P) =  \Phi_{s_0}(P'_i, P_{-i})$. Hence, when a mechanism is weakly non-bossy, no student can modify the set of {\it{unassigned students}} without changing her own school. Since DA satisfies weak non-bossiness (Bando and Imamura, 2016; Theorem 3), it is interesting to compare it with local non-bossiness.\\ 

\begin{lemma}\label{BaIma}
Weak non-bossiness is independent of both local non-bossiness and strategy-proofness.   
\end{lemma}

\begin{proof}
On the one hand, it is not difficult to verify that the mechanism described in \autoref{ex1_app} is weakly non-bossy. However, this mechanism is neither locally non-bossy nor strategy-proof. On the other hand, if $N=\{1,2,3\}$, $S=\{s_1,s_2\}$, and $(q_{s_1},q_{s_2})=(2,1),$  regardless of the schools' priority profile $(\succ_{s_1}, \succ_{s_2})$, the mechanism $\Omega:{\mathcal{P}}\rightarrow {\mathcal{M}}$ defined by
$$\Omega(P)=
\left\{
\begin{array}{ll}
((1,s_1),(2,s_1),(3,s_2)), &\quad\quad \mbox{when}\, \, s_2 P_1 s_1,\\
((1,s_1),(2,s_1),(3,s_0)), &\quad\quad \mbox{otherwise},
\end{array}
\right.
$$
 is locally non-bossy and strategy-proof but does not satisfy weak non-bossiness.
 \end{proof}

 \medskip

In one-to-one matching problems, local non-bossiness only requires that a student without a school cannot modify the set of unassigned students without getting a place elsewhere. Thus, weak non-bossiness is stronger than local non-bossiness in this context. Furthermore, although every group strategy-proof mechanism is weakly non-bossy, it follows from the proof of \autoref{BaIma} that the combination of local non-bossiness and strategy-proofness does not imply weak non-bossiness.\medskip

Group strategy-proofness is stronger than group non-bossiness (Afacan, 2012). However, the following example shows that the analogous result does not hold for the local versions of these concepts. \\

\begin{example}\label{ex2_app}{\it{A locally group strategy-proof mechanism that is locally bossy and weakly bossy}}.

Let $N=\{1,2,3\}$, $S=\{s_1,s_2\}$, $(q_{s_1},q_{s_2})=(2,1),$ $\succ_{s_1}:1, 2,3$, and $\succ_{s_2}:3, 1,2.$ 

Let $\Omega:{\mathcal{P}}\rightarrow {\mathcal{M}}$ be such that $\Omega(P)={\rm{DA}}(P)$ unless the most preferred school of student $1$ is $s_1$. In this case, let $\Omega_1(P)={s_1}$ and $\Omega_i(P)=s_0$ for all $i\neq 1$. Notice that, since $1$ has the highest priority at $s_1$, $\Omega_1(P) R_1 s_1$ for all $P\in {\mathcal{P}}$.

The mechanism $\Omega$ is strategy-proof, because ${\rm{DA}}$ satisfies this property and no one can prevent $1$ to receive a seat in $s_1$ when it is her most preferred alternative. Moreover, $\Omega$ is locally group strategy-proof. Indeed, if this were not the case, there would exist $P, P'\in {\mathcal{P}}$ and two students in $\Omega_{s_1}(P)=\{i,j\}$ such that if they report $(P_i', P_j')$ none of them would be worse off, and at least one of them strictly better off. Note that $s_1$ is not the most preferred school for student $1$ under either $P_1$ (because in this case only one student is  assigned to $s_1$) or $P'_1$ (because otherwise at least one of the students $\{i,j\}$ is not assigned under $(P'_i,P'_j, P_k)$ and her situation worsens). Thus, $\Omega$ coincides with ${\rm{DA}}$ in the preference profiles $P$ and $(P'_i,P'_j, P_k)$. This contradicts the local group strategy-proofness of ${\rm{DA}}$ (see \autoref{loc_nonb} and \autoref{papai}).

However, $\Omega$ is locally bossy. Consider the preference profile $P=(P_1,P_2,P_3)$:
\begin{flalign*}
P_1: s_2, s_1, s_0,\quad\quad\quad
P_2: s_1, s_2, s_0, \quad\quad\quad
P_3: s_2, s_1, s_0.
\end{flalign*}

We have that $\Omega(P)=((1,s_1),(2,s_1),(3,s_2))$ and $\Omega_{s_1}(P)=\{1,2\}$. If student $1$ changes her preference to $P'_1$ such that $P'_1:s_1, s_2,s_0$, she remains assigned to $s_1$ but is left without schoolmates because $\Omega_{s_1}(P'_1,P_2,P_3)=\{1\}$. This argument also shows that $\Omega$ is weakly bossy, because $\Omega_{s_0}(P)=\emptyset$ and $\Omega_{s_0}(P'_1,P_2,P_3)=\{2,3\}$. \hfill $\Box$\\
\end{example}

Although the school-optimal stable mechanism is not strategy-proof, Afacan and Dur (2017) show that it is non-bossy. Hence, \autoref{ex2_app} implies that local non-bossiness and strategy-proofness are independent properties.\\

The next two examples show that local group non-bossiness is stronger than local non-bossiness, and that local group strategy-proofness is stronger that strategy-proofness.\\

\begin{example}\label{ex3_app}{\it{A locally non-bossy mechanism that is not locally group non-bossy}}.

Let $N=\{1,2,3\}$, $S=\{s_1,s_2\}$, $(q_{s_1},q_{s_2})=(3,1),$ $\succ_{s_1}:1,2,3$, and $\succ_{s_2}: 3,2,1$. 

Denote by $top(P_i)$ the most preferred alternative of a student $i$ when her preferences are $P_i\in {\mathcal{L}}$. Let $\Omega:{\mathcal{P}}\rightarrow {\mathcal{M}}$ be the mechanism such that: $$\Omega(P)=
\left\{
\begin{array}{ll}
((1,s_1),(2,s_1),(3,s_2)), &\quad\quad \mbox{when}\, \, top(P_1)=top(P_2)=s_1;\\
((1,s_1),(2,s_1),(3,s_1)), &\quad\quad \mbox{when}\,\,  top(P_1)=top(P_2)=s_2;\\
((1,s_0),(2,s_0),(3,s_0)), &\quad\quad \mbox{otherwise}. 
\end{array}
\right.
$$

It is not difficult to verify that $\Omega$ is locally non-bossy. Indeed, $1$ and $2$ are the only two students who might change their schoolmates without changing schools. However, the definition of $\Omega$ ensures that, if $1$ or $2$ maintains her assignment after changing her preferences, then the whole matching does not change.

Moreover, $\Omega$ is locally group bossy. Let $P=(P_1,P_2,P_3)$ be a preference profile such that $top(P_1)=top(P_2)=s_1$ and consider $P'_1, P'_2 \in {\mathcal{L}}$ such that $top(P'_1)=top(P'_2)=s_2$. Notice that, although $\Omega_i(P)=\Omega_i(P'_1,P'_2,P_3)=s_1$ for each student $i\in \{1,2\}$, we have that $\Omega_{s_1}(P)=\{1,2\}\neq\{1,2,3\}=\Omega_{s_1}(P'_1,P'_2,P_3)$.\hfill $\Box$  \\
\end{example}

\begin{example}\label{ex4_app}{\it{A strategy-proof mechanism that is not locally group strategy-proof}}.

Let $N=\{1,2,3\}$, $S=\{s_1,s_2\}$, $(q_{s_1},q_{s_2})=(2,2),$ $\succ_{s_1}:1,2,3$, and $\succ_{s_2}:3,2,1$.  

Denote by $top_k(P_i)$ the $k$-th most preferred alternative of a student $i$ under $P_i\in {\mathcal{L}}$. Let $\Omega:{\mathcal{P}}\rightarrow {\mathcal{M}}$ be the mechanism such that:
$$\Omega(P)=
\left\{
\begin{array}{ll}
((1,top_1(P_1)),(2,s_1),(3,s_2)), &\quad\quad \mbox{when}\, \, top_2(P_1)=s_2;\\
((1,top_1(P_1)),(2,s_2),(3,s_1)), &\quad\quad \mbox{otherwise}.
\end{array}
\right.
$$

It is easy to verify that $\Omega$ is strategy-proof. We claim that $\Omega$ is not locally group strategy-proof. Let $P=(P_1,P_2,P_3)$ be a preference profile such that $P_1:s_1, s_2,s_0$, $P_2:s_2, s_1, s_0$, and $P_3: s_1,s_2,s_0$. Let $P'_1, P'_2 \in {\mathcal{L}}$ be such that $P'_1: s_1,s_0, s_2$ and $P'_2=P_2$. It follows that the coalition of schoolmates $\{1,2\}\in \Omega_{s_1}(P)$ can manipulate the mechanism $\Omega$ at $P$, because $\Omega_1(P)=s_1=\Omega_1(P'_1,P'_2,P_3)$ and $\Omega_2(P'_1,P'_2,P_3)=s_2 P_2 s_1=\Omega_2(P)$.\hfill $\Box$ \\
\end{example}

\section{On local non-bossiness under non-responsive preferences}\label{extension}\medskip

Given a set $N$ of students and a set $S$ of schools, consider a context in which every school $s$ has a complete, transitive, and strict preference relation $P_s$ defined on $2^N$ (the collection of subsets of $N$). Let $C^s:2^N \rightarrow 2^N$ be the mapping that associates to each $N'\subseteq N$ the most preferred subset of $N'$ under $P_s$. We refer to $C^s$ as the {\bf{choice function}} induced by $P_s$.

In this context, when students have preferences $(P_i)_{i\in N}\in {\mathcal{P}}$, a matching $\mu$ is stable as long as the following properties hold:
\begin{itemize}
\item[-] {\it{Individual rationality.}} For all $i\in N$ and $s\in S$, $\mu(i) R_i s_0$ and $C^s(\mu(s))=\mu(s)$.  
\item[-] {\it{No blocking.}} There is no $(i,s)\in N\times S$ such that $sP_i \mu(i)$ and $i\in C^s(\mu(s)\cup\{i\})$.
\end{itemize}

Notice that, throughout the paper, we implicitly assume that schools have {\bf{responsive preferences}} (Roth and Sotomayor, 1990), as the relevant information about how schools choose students from a pool of applicants is given by priority orders $(\succ_s)_{s\in S}$ and capacities $(q_s)_{s\in S}$. More precisely, when schools' preferences are responsive, for each $s\in S$ and $N'\in 2^N$, $C^s(N')$ is composed of the $\min\{\vert N'\vert, q_s\}$ students in $N'$ with the highest priority under $\succ_s$. 

Without responsiveness, the existence of a stable matching can be compromised. To avoid this possibility and to ensure that some classical properties of ${\rm{DA}}$ continue to hold,\footnote{The mechanism ${\rm{DA}}$ can be naturally adapted to more general contexts by assuming that each school uses its choice function to evaluate proposals when the student-proposing deferred acceptance algorithm is implemented.}  it is usually assumed that every $P_s$ satisfies the following conditions:
\begin{itemize}
\item {\bf{Substitutability}}: Given $N', N'' \in 2^N$ with $N'\subseteq N''$,  $$[\, i\in N',\,\,\,\quad i\in C^s(N'')\,]\quad\quad\Longrightarrow \quad\quad i\in C^s(N').$$ 
\item {\bf{Law of aggregate demand:}} Given $N', N''\in 2^N$ with $N'\subseteq N''$, $\vert C^s(N')\vert \leq \vert C^s(N'')\vert.$ \medskip
\end{itemize}

Under substitutability, if a student is chosen from a group of candidates, she will still be chosen from any subgroup of candidates. On the other hand, the law of aggregate demand ensures that the number of students chosen does not shrink when the pool of candidates expands. In every school choice context $[N,S,(P_s)_{s\in S}]$ in which $(P_s)_{s\in S}$ are substitutable and satisfy the law of aggregate demand, ${\rm{DA}}$ is the student-optimal stable mechanism and is strategy-proof (cf., Roth, 1984a; Hatfield and Milgrom, 2005).

It is natural to ask whether ${\rm{DA}}$ is locally non-bossy in this more general framework. Unfortunately, as the following example illustrates,  the local non-bossiness of ${\rm{DA}}$ cannot be ensured when we only know that schools' preferences are substitutable and satisfy the law of aggregate demand, even when only one school has non-responsive preferences.\\

\begin{example}\label{contraejemplo} Let $[N,S, (P_{s})_{s\in S}]$ be such that $N=\{1,2,3,4\}$, $S=\{s_1,s_2,s_3\}$,  
$$P_{s_1}: \{1,2\}, \{1,3\}, \{1,4\}, \{2,3\}, \{2,4\},\{1\},\{2\},\{3\}, {\underline{\{3,4\}}}, \{4\},\emptyset,\ldots,$$ 
and assume that $P_{s_2}$ and $P_{s_3}$ are responsive and compatible with the following priority orders and capacities: $\succ_{s_2}: 3,4, 1,2$, $\,\succ_{s_3}: 2,3,1,4$, and $q_{s_2}=q_{s_3}=1$. Notice that $P_{s_1}$ is substitutable and satisfies the law of aggregate demand. 

Consider the preference profile $P\in {\mathcal{P}}$ such that
\begin{eqnarray*}
P_1: s_2,s_1,s_3,s_0,  \quad\quad\,\,
P_2: s_2,s_3,s_1,s_0 \quad \quad \,\,
P_3: s_3,s_1, s_2,s_0\quad \quad\,\,
P_4: s_1, s_2,s_3,s_0  
\end{eqnarray*}
and let $P'_1: s_1,s_2, s_3, s_0$. It follows that 
\begin{eqnarray*}
\quad {\rm{DA}}(P) &=& ((1,s_1), (2,s_3), (3,s_1), (4,s_2)),\\
{\rm{DA}}(P'_1, P_{-1}) &=& ((1,s_1), (2,s_2), (3,s_3), (4,s_1)).
\end{eqnarray*}

Since ${\rm{DA}}_1(P)={\rm{DA}}_1(P'_1, P_{-1})$ and ${\rm{DA}}_{s_1}(P)\neq {\rm{DA}}_{s_1}(P'_1, P_{-1})$,   ${\rm{DA}}$ is locally bossy.\footnote{By raising the position of $\{3,4\}$ in $P_{s_1}$ we can ensure that ${\rm{DA}}$ is locally non-bossy. Indeed, the preference relation $$ \widehat{P}_{s_1}: \{1,2\}, \{1,3\}, \{1,4\}, \{2,3\}, \{2,4\},{\underline{\{3,4\}}}, \{1\},\{2\},\{3\}, \{4\},\emptyset,\ldots,$$ is responsive and compatible with the priority order $\succ_{s_1}:1,2,3,4$ and the capacity $q_{s_1}=2$. Hence, \autoref{loc_nonb} ensures that ${\rm{DA}}$ is locally non-bossy in the school choice context $[N,S, (\widehat{P}_{s_1}, P_{s_2}, P_{s_3})]$.}  \hfill $\Box$\\
\end{example}

To prove the local non-bossiness of ${\rm{DA}}$ (\autoref{loc_nonb}), we use the fact that any improvement cycle from ${\rm{DA}}$ must be blocked by some student. This property holds when schools have responsive preferences and consider all students admissible, because ${\rm{DA}}$ is the student-optimal stable mechanism and the reassignments induced by an improvement cycle never compromise individual rationality. However, this argument may fail in more general contexts.

For instance, in \autoref{contraejemplo}, although $(1,4)$ is an improvement cycle from ${\rm{DA}}(P)$ that is not blocked by any student, its existence does not contradict the optimality of ${\rm{DA}}$, because $C^{s_1}(\{3,4\})=\{3\}\neq \{3,4\}$ and, therefore, the matching obtained when it is implemented, $\eta=((1,s_2), (2,s_3), (3,s_1), (4,s_1))$, is not individually rational.

Therefore, we might think that the local bossiness of ${\rm{DA}}$ in \autoref{contraejemplo}  is driven by the fact that there is no $q_{s_1}\in {\mathbb{N}}$ such that $\vert C^{s_1}(N')\vert =\min\{q_{s_1}, \vert N'\vert \}$ for all $N'\in 2^N$. Indeed, $C^{s_1}(\{1,2\})=\{1,2\}$ and $C^{s_1}(\{3,4\})=\{3\}$ are incompatible with this property. To avoid this possibility, assume that every $P_s$ satisfies the following assumption: 
\begin{itemize}
\item {\bf{$\boldsymbol{q}$-acceptance.}} There exists $q_s\in {\mathbb{N}}$ such that $\vert C^s(N')\vert=\min\{q_s,\vert N'\vert\}$ for all $N'\in 2^N$.
\end{itemize}

Notice that $q$-acceptance is stronger than the law of aggregate demand. As the next example shows, even when schools preferences are substitutable and $q$-acceptant, ${\rm{DA}}$ may not be locally non-bossy.\\

\begin{example}\label{contraejemplo2} Let $[N,S, (P_{s})_{s\in S}]$ be such that $N=\{1,2,3,4,5\}$, $S=\{s_1,s_2,s_3,s_4\}$, and 
$$P_{s_1}: \{1,2\}, \{1,3\}, {\underline{\{1,5\}}}, {\underline{\{1,4\}}}, \{2,3\}, \{2,4\}, \{2,5\}, \{3,4\}, \{3,5\}, \{4,5\}, \{1\},\ldots,\{5\},\emptyset,\ldots,$$ 
Moreover, $P_{s_2}$, $P_{s_3}$, and $P_{s_4}$ are responsive preferences compatible with the following priority orders and capacities: $$\succ_{s_2}: 3,4,5, 1,2,\quad\quad\succ_{s_3}: 2,3,1,4,5,\quad\quad \succ_{s_4}: 1,2,3,4,5,\quad\quad q_{s_2}=q_{s_3}=q_{s_4}=1.$$ 

Notice that $P_{s_1}$ is substitutable and $q$-acceptant, with $q_{s_1}=2$. 

Consider the preference profile $P\in {\mathcal{P}}$ such that
$$
\begin{array}{l}
P_1: s_2,s_1,s_3,s_4,s_0,\\
P_2: s_2,s_3,s_1,s_4,s_0, \\
P_3: s_3,s_1, s_2,s_4,s_0,
\end{array} \quad\quad\quad\quad\quad
\begin{array}{l}
P_4: s_1, s_4,s_3,s_2,s_0, \\ 
P_5: s_1, s_2, s_4,s_3,s_0.\\
\,\,
\end{array}
$$
and let $P'_1: s_1,s_2, s_3, s_4, s_0$. It follows that 
\begin{eqnarray*}
\quad {\rm{DA}}(P) &=& ((1,s_1), (2,s_3), (3,s_1), (4,s_4), (5,s_2)),\\
{\rm{DA}}(P'_1, P_{-1}) &=& ((1,s_1), (2,s_2), (3,s_3), (4,s_4), (5,s_1)).
\end{eqnarray*}

Since ${\rm{DA}}_1(P)={\rm{DA}}_1(P'_1, P_{-1})$ and ${\rm{DA}}_{s_1}(P)\neq {\rm{DA}}_{s_1}(P'_1, P_{-1})$,   ${\rm{DA}}$ is locally bossy.\footnote{By swapping the positions of $\{1,5\}$ and $\{1,4\}$ in $P_{s_1}$ we can ensure that ${\rm{DA}}$ is locally non-bossy. Indeed, the preference relation $$ \widetilde{P}_{s_1}:\{1,2\}, \{1,3\}, {\underline{\{1,4\}}}, {\underline{\{1,5\}}}, \{2,3\}, \{2,4\}, \{2,5\}, \{3,4\}, \{3,5\}, \{4,5\}, \{1\},\ldots,\{5\},\emptyset,\ldots,$$ is responsive and compatible with the priority order $\succ_{s_1}:1,2,3,4,5$ and the capacity $q_{s_1}=2$. Hence, \autoref{loc_nonb} ensures that ${\rm{DA}}$ is locally non-bossy in the school choice context $[N,S, (\widetilde{P}_{s_1}, P_{s_2}, P_{s_3})]$.}  \hfill $\Box$\\
\end{example}

In the proof of \autoref{loc_nonb} (Case I) we use the fact that no student who keeps her preferences and assignment when moving from $\mu\equiv {\rm{DA}}(P)$ to $\mu'\equiv {\rm{DA}}(P'_1, P_{-1})$ blocks a $\mu$-improvement cycle.\footnote{A student $i$ {\bf{blocks a $\boldsymbol{\mu}$-improvement cycle}} if she blocks the matching $\eta$ obtained when this cycle is implemented from $\mu$. That is, when there exists $s\in S$ such that $sP_i \eta(i)$ and $i\in C^s(\eta(s)\cup\{i\})$.}  Indeed, when schools have responsive preferences, such a block is incompatible with the stability of $\mu'$ (see Footnote 13). \autoref{contraejemplo2} shows that this incompatibility disappears in more general contexts: student $4$ blocks the $\mu$-improvement cycle $(1,5)$, but she does not block $\mu'$. The former property is a consequence of $C^{s_1}(\{3,4,5\})=\{3,4\}$, while the latter follows from $C^{s_1}(\{1,4,5\})=\{1,5\}$. \\

\section{proof of theorem 2}\label{appT2} \medskip

In this appendix, we present the proof of \autoref{coro-bando}. First, we formally define the concept of a choice function and show how one can be derived for each school from a mechanism $\Phi$. Second, we use the results of Chambers and Yenmez (2018) to show that, for every mechanism satisfying the axioms of \autoref{coro-bando}, the choice function derived for a school $s$ is compatible with a choice based on a priority order $\succ_s$. Finally, we use these priority orders $\succ=(\succ_s)_{s\in S}$ to guarantee that $\Phi = {\rm DA}^\succ$.\\

A {\bf{choice function}} is a mapping $C:{\mathcal{N}}\rightarrow {\mathcal{N}}$ such that, for each $N\in {\mathcal{N}}$, $C(N)$ is a {\it{non-empty}} subset of $N$.\footnote{Since schools' preferences are represented by priority orders that consider all students admissible, requiring a choice function that is defined on ${\mathcal{N}}$ to have  non-empty values does not limit generality.} 
Given a choice function $C$, consider the following properties:
\begin{itemize}
\item $C$ is {\bf{${\boldsymbol{q}}$-responsive}} when there is a priority order $\succ$ defined on $\overline{N}$ such that:
\begin{itemize}
\item For every $N\in {\mathcal{N}}$ with $\vert N\vert > q$, $C(N)$ is the set of $q$ students in $N$ with highest priority under $\succ$.
\item For every $N\in {\mathcal{N}}$ with $\vert N\vert \leq q$, $C(N)=N$.
\end{itemize}
\item $C$ is {\bf{substitutable}} when for every $N, N'\in {\mathcal{N}}$ with $N\subseteq N'$, $$[\,i\in N, \quad i\in C(N')\, ]\quad\quad\Longrightarrow \quad\quad i\in C(N).$$ 
\item $C$ satisfies the $\boldsymbol{(q+1)}${\bf{-weaker axiom of revealed preference}} ({\bf{$\boldsymbol{(q+1)}$-WrARP}}) when for every $i,j\in \overline{N}$ and $N,N'\in {\mathcal{N}}$ such that $i,j\in N\cap N'$ and $\vert N\vert = \vert N' \vert =q+1$, $$[\,i\in C(N),\quad j\in C(N')\setminus C(N)\,]\quad\quad\Longrightarrow\quad\quad i \in C(N').$$
\item $C$ is {\bf{$\boldsymbol{q}$-acceptant}} when for every $N\in {\mathcal{N}}$, $\vert C(N)\vert= \min\{q, \vert N\vert\}.$\\
\end{itemize}

{\bf{Throughout this appendix, for each $\boldsymbol{s\in S}$, let $\boldsymbol{P^s}\in {\mathcal{L}}$ be a preference relation such that $\boldsymbol{s}$ is the only admissible school.}} 

Given a mechanism $\Phi:{\mathcal{N}}\times \mathcal{P}\rightarrow {\mathcal{M}}^*$, denote by $C^{\Phi,s}:{\mathcal{N}}\rightarrow {\mathcal{N}}$ the function defined by $$C^{\Phi,s}(N)=\Phi_s(N,(P^s,\ldots, P^s)),\quad\quad\forall N\in {\mathcal{N}}.$$ 

Note that, although there are many ways to choose the preference relation $P^s$, the fact that for every $P\in {\mathcal{P}}$ the matching $\Phi(N,P)$ depends only on $[N,S,q,(P^t_i)_{i\in N}]$ guarantees that the mapping $C^{\Phi,s}$ does not depend on how schools different from $s$ are ordered in $P^s$.\\

\begin{lemma}\label{NUEVO}
Given an individual rational and weakly non-wasteful mechanism $\Phi:{\mathcal{N}}\times \mathcal{P}\rightarrow {\mathcal{M}}^*$, the following properties hold for each $s\in S$:
\begin{itemize}
\item[(i)] $C^{\Phi,s}$ is a $q_s$-acceptant choice function.
\item[(ii)] $C^{\Phi,s}$ is substitutable whenever $\Phi$ is population-monotonic.
\item[(iii)] $C^{\Phi,s}$ satisfies $(q_s+1)$-WrARP for all $s\in S$ if and only if $\Phi$ satisfies S-WrARP.\\
\end{itemize} 
\end{lemma}

\begin{proof} {\bf{(i)}} Since $s$ is the only admissible school under $P^s$, the individual rationality and weak non-wastefulness of $\Phi$ ensures that, for every $N\in {\mathcal{N}}$, the following properties hold:
\begin{itemize}
\item If $\vert N \vert \leq  q_s$, then $\Phi_s(N, (P^s,\ldots, P^s))=N$.
\item If $\vert N \vert >  q_s$, then  $\vert \Phi_s(N, (P^s,\ldots, P^s))\vert = q_s$.
\end{itemize}
Hence, $\Phi_s(N, (P^s,\ldots, P^s))$ is non-empty and $\vert \Phi_s(N, (P^s,\ldots, P^s))\vert=\min\{q_s, \vert N\vert\}$. Therefore, it follows from the definition of $C^{\Phi,s}$ that it is a $q_s$-acceptant choice function, as $C^{\Phi,s}(N)\neq \emptyset$ and $\vert C^{\Phi,s}(N)\vert=\min\{q_s, \vert N\vert\}$ for all $N\in {\mathcal{N}}$. \medskip

\noindent {\bf{(ii)}} Let $N,N'\in {\mathcal{N}}$ be such that $N\subseteq N'$ and $i\in C^{\Phi,s}(N')\cap N$. Then, $i\in \Phi_s(N', (P^s,\ldots, P^s))$ and $i\in N$. Since $s$ is the only admissible school under $P^s$, the population-monotonicity of $\Phi$ implies that $i\in \Phi_s(N, (P^s,\ldots, P^s))$. That is, $i\in C^{\Phi,s}(N)$. We conclude that $C^{\Phi,s}$ is substitutable.\medskip

\noindent {\bf{(iii)}} Let $N\in {\mathcal{N}}$, $P,\tilde{P}\in {\mathcal{P}}$, and $s\in S$ be such that $A(P_i)=A(\tilde{P}_i)=\{s\}$ for all $i\in N$. From the definition of mechanism, we know that $\Phi_s(N,P)$ depends only on $[N,S, q, (P^t_i)_{i\in N}]$ and $\Phi_s(N,\tilde{P})$ depends only on $[N,S, q, (\tilde{P}^t_i)_{i\in N}]$, where $P^t_i$ and $\tilde{P}^t_i$ are the truncations of $P_i$ and $\tilde{P}_i$ at $s_0$, respectively. Since $(P^t_i)_{i\in N}=(\tilde{P}^t_i)_{i\in N}$, it follows that $\Phi_s(N,P)=\Phi_s(N,\tilde{P})$. Therefore, it is immediate from the definitions of $(q_s+1)$-WrARP and S-WrARP that the choice function $C^{\Phi,s}$ satisfies $(q_s+1)$-WrARP for all $s\in S$ if and only if $\Phi$ satisfies S-WrARP.\end{proof}

$\,$

The next result is an immediate consequence of Theorem 2 of Chambers and Yenmez (2018) and \autoref{NUEVO}  above:\\

\begin{lemma}\label{NUEVO-1}
If a mechanism $\Phi:{\mathcal{N}}\times \mathcal{P}\rightarrow {\mathcal{M}}^*$ is individually rational, weakly non-wasteful, population-monotonic, and satisfies S-WrARP, then $C^{\Phi,s}$ is $q_s$-responsive for all $s\in S$.\\
\end{lemma}

The following result shows that, when an individually rational and strategy-proof mechanism $\Phi$ is implemented in a population $N\subseteq {\overline{N}}$, if a student $i$ is assigned to a school $s$ under $P$, then she continues to be assigned to $s$ when she reports it as her only admissible alternative. Furthermore, if $s$ is unattainable for $i$ under $P$, in the sense that $s P_i \Phi_i(N,P)$, then it is unattainable for her even when she reports $s$ as her only admissible school.\\

\begin{lemma}\label{NUEVO-2}
Let $\Phi:{\mathcal{N}}\times \mathcal{P}\rightarrow {\mathcal{M}}^*$ be an individually rational and strategy-proof mechanism. For every $N\in {\mathcal{N}}$, $P\in {\mathcal{P}}$, $s\in S$, and $i\in N$ the following properties hold:
\begin{itemize}
\item[(i)] If $\Phi_i(N,P)=s$, then $\Phi_i(N, (P^s, P_{-i}))=s$.
\item[(ii)] If $s P_i \Phi_i(N,P)$, then $\Phi_i(N, (P^s, P_{-i}))=s_0$.\\
\end{itemize}
\end{lemma}

\begin{proof} {\bf{(i)}} Suppose that $\Phi_i(N,P)=s$. Note that the definition of $P^s$ and the individual rationality of $\Phi$ imply that $\Phi_i(N,(P^s,P_{-i})) \in \{s,s_0\}$. Moreover, if $\Phi_i(N,(P^s, P_{-i}))=s_0$, then $\Phi_i(N,P)=s \ P^s \ s_0= \Phi_i(N,(P^s,P_{-i}))$, which contradicts the strategy-proofness of $\Phi$. Therefore, $\Phi_i(N,(P^s, P_{-i}))=s$. \medskip

\noindent {\bf{(ii)}} Suppose that  $s P_i \Phi_i(N,P)$. If $\Phi_i(N,(P^s, P_{-i})) = s$, then   $\Phi_i(N,(P^s,P_{-i})) \ P_i \ \Phi_i(N,P)$, which contradicts the strategy-proofness of $\Phi$. Thus, as $\Phi_i(N,(P^s,P_{-i})) \in \{s,s_0\}$, it follows that $\Phi_i(N,(P^s,P_{-i}))=s_0$.
\end{proof}

$\,$

\noindent {\bf{Proof of Theorem 2.}} Given a priority profile $\succ$, the stability of ${\rm{DA}}^{\succ}$ ensures that it is individually rational and weakly non-wasteful. ${\rm{DA}}^{\succ}$ is also population-monotonic (Crawford, 1991), strategy-proof (Dubins and Freedman, 1981; Roth, 1982), and satisfies S-WrARP because under deferred acceptance each school chooses the students according to its priority order. The weak local non-bossiness of ${\rm{DA}}^{\succ}$ follows from \autoref{loc_nonb}. 

Suppose that $\Phi:{\mathcal{N}}\times \mathcal{P}\rightarrow {\mathcal{M}}^*$ satisfies individual rationality (IR), weak non-wastefulness (WNW), population-monotonicity (PM), strategy-proofness (SP), S-WrARP, and weak local non-bossiness (WLNB). It follows from \autoref{NUEVO-1} that, for every $s\in S$, the mapping that associates to each $N\in {\mathcal{N}}$ the set $\Phi_s(N,(P^s,\cdots, P^s))$ is a $q_s$-responsive choice function. Hence, there exists a priority order $\succ_s$ defined on $\overline{N}$ such that:
\begin{itemize}
\item For every $N\in {\mathcal{N}}$ such that $\vert N\vert > q_s$, $\Phi_s(N,(P^s,\cdots, P^s))$ is the set of $q_s$ students in $N$ with higher priority under $\succ_s$.
\item For every $N\in {\mathcal{N}}$ such that $\vert N\vert \leq q_s$, $\Phi_s(N,(P^s,\cdots, P^s))=N$.
\end{itemize}

In particular, the following property holds:
\begin{quote}
\quad {\bf{(A)}} For every $s\in S$ and $N\in {\mathcal{N}}$, $$\quad\quad\quad \quad i\in N\setminus\Phi_s(N,(P^s,\cdots, P^s))\,\,\,\quad\Longrightarrow\,\,\,\quad  j\succ_s i,\,\,\,\forall  j\in \Phi_s(N,(P^s,\cdots, P^s)).$$
\end{quote}

Let $\succ=(\succ_s)_{s\in S}$. Notice that ${\rm{DA}}^\succ$ satisfies a property analogous to (A):
\begin{quote}
\quad {\bf{(B)}} For every $s\in S$ and $N\in {\mathcal{N}}$, $$\quad\quad\quad \quad i\in N\setminus {\rm{DA}}^\succ_s(N,(P^s,\cdots, P^s))\,\,\,\quad\Longrightarrow\,\,\,\quad  j\succ_s i,\,\,\,\forall  j\in {\rm{DA}}^\succ_s(N,(P^s,\cdots, P^s)).$$
\end{quote}

We will use properties (A) and (B) to show that $\Phi={\rm{DA}}^{\succ}$. 

By contradiction, assume that there is $N^*\in {\mathcal{N}}$ such that $\Phi(N^*,\cdot)\neq {\rm{DA}}^{\succ}(N^*,\cdot)$. \medskip

For each $P\in {\mathcal{P}}$, let $Z(P)=\vert \{i\in N^*: \vert A(P_i)\vert \leq 1 \}\vert $ be the number of students in $N^*$ who find at most one school admissible. Let $P^*\in {\mathcal{P}}$ be such that $\Phi(N^*,P^*)\neq {\rm{DA}}^{\succ}(N^*,P^*)$ and $Z( P^*)\geq Z(P)$ for all $P\in {\mathcal{P}}$ satisfying  $\Phi(N^*,P)\neq {\rm{DA}}^{\succ}(N^*,P)$. {Moreover, as $P^*$ is not necessarily uniquely defined, for each $i\in N^*$ choose $P^*_i$ from the set $\{P^s\}_{s\in S}$ when it is possible.\\

\noindent {\bf{Claim I.}} {\it{For each $j\in N^*$ such that $\Phi_j(N^*,P^*)\neq {\rm{DA}}^{\succ}_j(N^*,P^*)$ there is $s\in S$ such that $P^*_j=P^s$.}}\medskip

\noindent {\it{Proof.}} There are two cases to analyze:
\begin{itemize}
\item Suppose that $\Phi_j(N^*,P^*)\,P^*_j \,{\rm{DA}}^{\succ}_j(N^*,P^*)$. Since ${\rm{DA}}^{\succ}$ satisfies IR, $s\equiv \Phi_j(N^*,P^*)\in S$. Hence, $\vert A(P^*_j)\vert \geq  1$. Moreover, it follows from \autoref{NUEVO-2} that $\Phi_j(N^*,(P^s,P^*_{-j}))=s$ and ${\rm{DA}}^{\succ}_j(N^*,(P^s,P^*_{-j}))=s_0$. 
Since $\Phi(N^*,(P^s,P^*_{-j}))$ and ${\rm{DA}}^{\succ}(N^*,(P^s,P^*_{-j}))$ are different, if $\vert A(P^*_j)\vert >1$ then $Z(P^s, P^*_{-j})>Z(P^*)$ and we contradict the definition of $P^*$. Hence, $\vert A(P^*_j)\vert =  1$.
\item Suppose that ${\rm{DA}}^{\succ}_j(N^*,P^*)\,P^*_j \,\Phi_j(N^*,P^*) $. Then, following the arguments applied in the previous case but swapping the roles of $\Phi$ and ${\rm{DA}}^{\succ}$, we can conclude $\vert A(P^*_j)\vert =  1$.
\end{itemize}

In any of the cases above, there exists a school $s$ such that $A(P^*_j)=\{s\}$. Therefore, it follows from the definition of $P^*$ that $P^*_j=P^s$. This completes the proof of this claim.\\

Let $i\in N^*$ be such that $\Phi_i(N^*,P^*)\neq {\rm{DA}}^{\succ}_i(N^*,P^*)$. It follows from Claim I that there exists $s\in S$ such that $P^*_i=P^s$. Moreover, IR guarantees that $$s={\rm{DA}}^{\succ}_i(N^*,P^*)\,P^*_i \,\Phi_i(N^*,P^*)=s_0,\quad\quad\mbox{or}\quad\quad s=\Phi_i(N^*,P^*)\,P^*_i \, {\rm{DA}}^{\succ}_i(N^*,P^*)=s_0.$$

There are two scenarios to analyze:
\begin{itemize}
\item When $s={\rm{DA}}^{\succ}_i(N^*,P^*)\,P^*_i \,\Phi_i(N^*,P^*)=s_0$, WNW implies that $\vert \Phi_s(N^*,P^*)\vert=q_s$. Let $\Phi_s(N^*,P^*)\equiv\{i_1,\ldots,i_{q_s} \}$, where $i_1\succ_s \cdots\succ_s i_{q_s}$. Since $\Phi_{i_1}(N^*,P^*)=s$, it follows from \autoref{NUEVO-2} that $\Phi_{i_1}(N^*,(P^s,P^*_{-i_1}))=s$. As a consequence, WLNB ensures that
$$\Phi_s(N^*,(P^s,P^*_{-i_1}))=\Phi_s(N^*,P^*).$$

Applying the previous argument sequentially to $i_2,\ldots, i_{q_s}$, we conclude that
$$
\Phi_s(N^*,P')=\Phi_s(N^*,P^*),
$$
where the preference profile $P'$ is characterized by $P'_j=P^s$ for all $j\in \Phi_s(N^*,P^*)$ and $P'_j=P^*_j$ for all $j\in N^*\setminus \Phi_s(N^*,P^*)$.\footnote{This is the only step in the proof of \autoref{coro-bando} where WLNB is required (see \autoref{rem-thm2}).} Since $P^*_i=P^s$, if $N'\equiv \Phi_s(N^*, P^*)\cup\{i\}$, it follows that $P'_{N'}=(P^s,\ldots, P^s)$ and PM  implies that
$$
\Phi_s(N',P')=\{i_1,\ldots,i_{q_s} \}.
$$

Therefore, $\Phi_i(N',P')\neq s$ and property (A) imply that $i_{q_s}\succ_s i$.\footnote{Notice that, $\Phi_s(N',P')=\Phi_s(N', (P^s,\ldots, P^s))$. Indeed, the definition of a mechanism ensures that, for every $N\subseteq \overline{N}$ and $P\in {\mathcal{P}}$, the matching  $\Phi(N,P)$ depends only on $[N,S,q,(P^t_i)_{i\in N}]$, where $P^t_i$ is the truncation of $P_i$ at $s_0$.}

On the other hand, Claim I implies that every $j\in \Phi_s(N^*,P^*)$ is either assigned to $s$ under ${\rm{DA}}^\succ(N^*,P^*)$ or has $P^*_j=P^s$. Moreover, as $i\in {\rm{DA}}^{\succ}_s(N^*,P^*)$, there exists a student $k\in \Phi_s(N^*,P^*)$ such that $k\notin {\rm{DA}}^{\succ}_{s}(N^*,P^*)$. Hence, $P^*_k=P^s$ and    ${\rm{DA}}^{\succ}_{k}(N^*,P^*)=s_0$, which thanks to WNW ensures that $\vert{\rm{DA}}^\succ_s(N^*,P^*)\vert=q_s$. Denote $ {\rm{DA}}^\succ_s(N^*,P^*)\equiv\{i, i'_2, \ldots,i'_{q_s} \}$. We can apply a sequential argument analogous to the one used above, but now focused on students in ${\rm{DA}}^\succ_s(N^*,P^*)$ instead of $\Phi_s(N^*,P^*)$, in order to guarantee that, if $\hat{N}\equiv {\rm{DA}}^\succ_s(N^*,P^*)\cup\{k\}$, then
$${\rm{DA}}^\succ_s(\hat{N}, \hat{P})=\{i, i'_{2},\ldots, i'_{q_s}\},$$ 
where $\hat{P}$ is characterized by $\hat{P}_j=P^s$ for all $j\in {\rm{DA}}^\succ_s(N^*,P^*)$ and $\hat{P}_j=P^*_j$ for all $j\in N^*\setminus {\rm{DA}}^\succ_s(N^*,P^*)$. Hence, ${\rm{DA}}^\succ_k(\hat{N},\hat{P})\neq s$, and given that $\hat{P}_k= P^*_k=P^s$, it follows from property (B) that $i\succ_s k$. Since $k\in \Phi_s(N^*,P^*)$, we conclude that $i\succ_s i_{q_s}$. A contradiction.

\item When $s=\Phi_i(N^*,P^*)\,P^*_i \, {\rm{DA}}^{\succ}_i(N^*,P^*)=s_0,$ we can obtain a contradiction applying the arguments of the previous scenario but swapping the roles of $\Phi$ and ${\rm{DA}}^{\succ}$.
\end{itemize} 

It follows that the mechanisms 
$\Phi$ and ${\rm{DA}}^{\succ}$ coincide.\hfill $\Box$\\

\noindent 
\begin{remark} \label{rem-thm2}
{\textnormal{The sequential arguments applied in the proof of \autoref{coro-bando}, those where the preferences of a group of students are altered one by one without modifying the assignment, are crucial to adapt Ehlers and Klaus (2016, Theorem 1) to many-to-one matching problems. These arguments use WLNB and are applied as part of a strategy to obtain a contradiction after assuming that $\Phi_s(N^*,P^*)\neq {\rm{DA}}^{\succ}_s(N^*,P^*)$ for some $s\in S$ and $(N^*,P^*)\in {\mathcal{N}}\times {\mathcal{P}}$. However, as Claim I ensures that $P^*_j=P^s$ for all $j\in \Phi_s(N^*,P^*)\setminus {\rm{DA}}^{\succ}_s(N^*,P^*)$, 
the same techniques of Ehlers and Klaus (2016) could be applied to prove \autoref{coro-bando} if we knew that $ \Phi_s(N^*,P^*)\cap {\rm{DA}}^{\succ}_s(N^*,P^*)$ is an empty set. As expected, the latter always occurs in one-to-one matching markets: when $q_s=1$, $\Phi_s(N^*,P^*)\neq {\rm{DA}}^{\succ}_s(N^*,P^*)$ implies that $\Phi_s(N^*,P^*)\cap {\rm{DA}}^{\succ}_s(N^*,P^*)=\emptyset$. In summary, the axiom WLNB is required only because $ \Phi_s(N^*,P^*)\cap {\rm{DA}}^{\succ}_s(N^*,P^*)$ can be a non-empty set, and it  is something we cannot rule out in many-to-one matching markets simply by assuming that $\Phi_s(N^*,P^*)\neq {\rm{DA}}^{\succ}_s(N^*,P^*)$.}}\hfill $\Box$\\
\end{remark}

\noindent {\bf{Independence of the axioms in Theorem 2.}} The independence of individual rationality, weak non-wastefulness, population monotonicity, and strategy-proofness is shown in Ehlers and Klaus (2016) through several examples in one-to-one matching markets. Since these examples satisfy S-WrARP and weak local non-bossiness, they also show that there exist mechanisms that satisfy all the axioms of \autoref{coro-bando} except for one of the following: individual rationality, weak non-wastefulness, population monotonicity, or strategy-proofness. Moreover, our \autoref{EK-1} provides a scenario where all the axioms of \autoref{coro-bando} but weak local non-bossiness are satisfied. Hence, to ensure the independence of the axioms in \autoref{coro-bando}, it suffices to show that there are mechanisms that satisfy all the axioms but S-WrARP. \\

\begin{example} {\it{All the axioms are satisfied but S-WrARP.}}

Let $\bar N=\{1,2,3,4\}$, $S=\{s\}$ and $q_s=2$. Consider the profile $\tilde P$ such that $s\tilde P_i s_0$ for all $i\in \overline{N}$ and the mechanism $\Phi$ such that:
\begin{itemize}
\item $
    \Phi_s(\{1,2,3\}, \tilde P)=\{1,3\}.$
\item For $(N,P)\neq (\{1,2,3\}, \tilde{P})$, the assignment $\Phi_s(N,P)$ is computed using the serial dictatorship mechanism according to the order $(4,3,2,1)$, denoted by ${\rm{SD}}_{(4,3,2,1)}$.\footnote{The serial dictatorship mechanism assigns schools' seats following a given order of students: the first one chooses her favorite school, the next one chooses her favorite alternative among the schools that still have places unassigned, and so on.} 
\end{itemize}

Since ${\rm{SD}}_{(4,3,2,1)}$ coincides with ${\rm{DA}}^\succ$ when $\succ_s: 4,3,2,1$, it follows that it is individually rational, weakly non-wasteful, population-monotonic, strategy-proof, weakly locally non-bossy, and satisfies S-WrARP (\autoref{coro-bando}). Hence, it is straightforward to verify that $\Phi$ is individually rational and weakly non-wasteful. Furthermore, strategy-proofness and weak local non-bossiness trivially hold, because $P_i:s,s_0$ and $P'_i:s_0,s$ are the only preference relations that a student $i$ may have. 

To verify population monotonicity, we first compare the assignments $\Phi(\bar N, \tilde P)$ and $\Phi(\{1,2,3\}, \tilde P)$. In $\Phi(\bar N, \tilde P)$, students $3$ and $4$ are assigned to $s$, while in $\Phi(\{1,2,3\}, \tilde P)$ students $1$ and $3$ are  assigned to $s$. Since $4\notin \{1,2,3\}$, population-monotonicity holds when we compare $\overline{N}$ with $\{1,2,3\}$. On the other hand, if we 
begin with the subset \(\{1,2,3\}\) and consider a subpopulation $N\subset \{1,2,3\}$,
$\Phi$ assigns each student in $N$ to $s$. Therefore, none of them worsen their situation in relation with $\Phi(\{1,2,3\},\tilde{P})$. Hence, population-monotonicity holds.

Finally, by taking $N=\{1,2,3\}$, $N'=\{1,2,4\}$, $i=1$, and $j=2$, one can easily verify that \(\Phi\) does not satisfy S-WrARP at preference profile $\tilde{P}$. 
\hfill $\Box$\\
\end{example}
}

\section{omitted proofs and further remarks}\label{app} \medskip

Given a student $1\in N$ and preference profiles $P, P'\in {\mathcal{P}}$, let $\mu={\rm{DA}}(P_1,P_{-1})$ and $\mu'={\rm{DA}}(P'_1,P_{-1})$. The following result formalizes the claim made in \autoref{r1}. \\

\begin{lemma}\label{DA3}
If $\mu(1)\neq \mu'(1)$, then ${\rm{Coll}}_1(P) \cap {\rm{Coll}}_1(P'_1, P_{-1})=\emptyset$. \medskip
\end{lemma}

\begin{proof}

Let $\overline{s}\equiv \mu(1)\in S\cup\{s_0\}$ and $\hat{s}\equiv \mu'(1)\in S\cup\{s_0\}$. By contradiction, suppose that there exists $i_1 \in {\rm{Coll}}_1(P) \cap {\rm{Coll}}_1(P'_1, P_{-1})$. Hence, $i_1\neq 1$, $\mu(i_1)= \overline{s}$, and $\mu'(i_1)= \hat{s}$. Since $\overline{s}\neq \hat{s}$, we can assume that $\hat{s} P_{i_1}\overline{s}$ (in other case, exchange the roles of $P$ and $(P'_1, P_{-1})$). 

Let $I$ be the set of students who prefer $\mu'$ to $\mu$ under $P$. Note that $i_1\in I$. Consider the graph $G^*$  with nodes $V^*=\{i\in N: \,\, \mu(i)\in \mu(I)\}$ and directed edges $E^*=\{[i,j]: i\in I,\, \mu(j)=\mu'(i)\}.$ Hence, the nodes of $G^*$ are the students who are assigned in $\mu$ to schools that have at least one student of $I$ in the matching $\mu$, and there is an edge between every $i\in I$ and all the students who are assigned to $\mu'(i)$ in $\mu$. Since DA is  stable, each student in $I$ is assigned in $\mu'$ to a school that fills its quota at $\mu$. This implies that $G^*$ has a non-empty set of edges.

Applying to $G^*$ the same edge replacement process applied to the graph $G$ in the proof of \autoref{loc_nonb}---but considering $(V^*, E^*)$ instead of $(V,E)$ and using the fact that, when moving from $\mu$ to $\mu'$, every $i\in I$ displaces someone who was assigned to $\mu'(i)$ at $\mu$---we can construct a graph $G^{**}=(V^*, E^{**})$ with contains a $\mu$-improving cycle $(i^*_1,\ldots, i^*_r)$ that cannot be blocked by a student in $V^*$. 

We claim that $(i^*_1,\ldots, i^*_r)$ cannot be blocked by students in $N\setminus V^*$. Indeed, as $1\in V^*$, any $h\in N\setminus V^*$ has the same preferences in $P$ and $(P'_1, P_{-1})$ and consider $\mu(h)$ at least as preferred as $\mu'(h)$. Hence, if $h$ blocks $(i^*_1,\ldots, i^*_r)$, she will also block $\mu'$ (the arguments in Footnote 13 still hold when $\mu(h) P_h \mu'(h)$ instead of $\mu(h)=\mu'(h)$, and $E\setminus E'$ is replaced by $E^*\setminus E^{**}$).   

Therefore, $(i^*_1,\ldots, i^*_r)$ cannot be blocked by any student. A contradiction with the fact that $\mu$ is the student optimal stable matching under $P$.
\end{proof}

\medskip

The following example shows that the school-median stable mechanism introduced by Klaus and Klijn (2006) does not satisfy local non-bossiness.\\ 

\begin{example}\label{median}{\it{The school-median stable mechanism is locally bossy.}}

Given a school choice problem $[N,S,\succ, q,P]$ with $k$ stable matchings, we refer to $\mu$ as the {\bf{school-median stable matching}} when the following properties hold:
\begin{itemize}
\item If $k$ is odd, each student is assigned to her $(k+1)/2$-th (weakly) best match among all $k$ stable matchings.
\item If $k$ is even, each student is assigned to her $(k+2)/2$-th (weakly) best match among all $k$ stable matchings.
\end{itemize}

Klaus and Klijn (2006, Theorem 3.2) ensure that this matching always exists.

We claim that the mechanism that associates to each preference profile the school-median stable matching is locally bossy. Let $N=\{1,2,3,4\}$, $S=\{s_1,s_2\}$, $q_{s_1}=q_{s_2}=2$, $\succ_{s_1}: 2,1,3,4$, and $\succ_{s_2}: 3,4,2,1.$ If preferences are given by $
 P_1: s_2,s_1,s_0$, $P_2:s_2,s_1,s_0$, $P_3:s_1,s_2,s_0$, and $P_4: s_1,s_2,s_0$, then
there are three stable matchings:
\begin{eqnarray*}
\mu&=&((1,s_2), (2,s_2), (3,s_1), (4,s_1)),\\
\eta&=&((1,s_1), (2,s_2), (3,s_1), (4,s_2)),\\
\rho&=&((1,s_1), (2,s_1), (3,s_2), (4,s_2)).
\end{eqnarray*}

It follows that $\eta$ is the school-median stable matching. Now, when student $1$'s preferences change to $P'_1: s_1,s_2,s_0$, the only stable matchings are $\eta$ and $\rho$. In this case, $\rho$ is the school-median stable matching.
Therefore, as $\eta(1)=\rho(1)=s_1$ and $\eta(s_1) \neq \rho(s_1)$, we conclude that the school-median stable mechanism is locally bossy.\hfill $\Box$\\
\end{example}

The next example shows that the construction of  the graph $G'$ in the proof of \autoref{loc_nonb} is not superfluous. Indeed, the $\mu$-improving cycle $(i^*_1,\ldots, i^*_r)$, that allows us to show that $\mu(\overline{s})=\mu'(\overline{s}),$ may not be present in the graph $G$.\\

\begin{example}\label{ex-cycles} Let $N=\{1,2,3,4,5,6\}$, $S=\{s_1,s_2,s_3,s_4,s_5\}$, $q_{s_1}=2$, $q_s=1$ for every $s\neq s_1$. and 
$$\quad \succ_{s_1}: 6,1,\ldots,\quad\quad 
\succ_{s_2}: 5,1,3,2\ldots,\quad\quad
\succ_{s_3}: 4,2,3\ldots\quad\quad
\succ_{s_4}: 3,4\ldots\quad\quad
\succ_{s_5}: 2,5\ldots$$

Consider a preference profile $P\in {\mathcal{P}}$ such that
\begin{eqnarray*}
\begin{array}{ll}
P_1:& s_2,s_1,\ldots,  \\
P_2:& s_2,s_3,s_5,\ldots, 
\end{array}
& \quad\quad \quad 
\begin{array}{ll}
P_3:& s_3,s_2,s_4,\ldots,\\
P_4:& s_4, s_3,\ldots, 
\end{array} & \quad\quad\quad 
\begin{array}{ll}
P_5:& s_5,s_2,\ldots, \\
P_6:&  s_1,\ldots,
\end{array}
\end{eqnarray*}
and let $P'_1:s_1,\ldots$. Then 
\begin{eqnarray*}
\mu&\equiv& {\rm{DA}}(P)=((1,s_1), (2,s_5), (3,s_4), (4,s_3), (5,s_2), (6,s_1)),\\
\mu'&\equiv& {\rm{DA}}(P'_1, P_{-1})=((1,s_1), (2,s_2), (3,s_3), (4,s_4), (5,s_5), (6,s_1)).
\end{eqnarray*}

Using the notation from the proof of \autoref{loc_nonb}, it follows that $G=(V,E)$ is defined by $V=\{2,3,4,5\}$ and $E=\{[2,5], [3,4], [4,3], [5,2]\}$. Moreover, the graph $G'=(V,E')$ satisfies $E'=\{[3,5], [2,4], [4,3], [5,2]\}.$ Hence, $G'$ has a unique $\mu$-improving cycle, namely: $(2,4,3,5)$. As expected, student $1$ is the only student who $\mu$-blocks this cycle (at edge $[3,5]$). On the other hand, $G$ has two $\mu$-improving cycles: $(2,5)$ and $(3,4)$. The first one is $\mu$-blocked by the students $1$ and $3$ at edge $[2,5]$, while the second one is $\mu$-blocked by the student $2$ at edge $[3,4]$. More importantly, $\mu'$ does not implement $(2,4,3,5)$. \hfill $\Box$\\
\end{example}

In \autoref{ex-cycles}, student $1$ has a relevant role in the inefficiency of ${\rm{DA}}(P)$. Indeed, although $1$ is assigned to $s_1$, she displaces $2$ from $s_2$. However, the displacement of $2$ induces other inefficiencies: although $2$ is assigned to $s_5$, she displaces $3$ from $s_3$. Hence, when student $1$ reports $P'_1$, the improvement in welfare with respect to ${\rm{DA}}(P)$ is greater than the one generated when only $(2,4,3,5)$ is implemented. Indeed, the matching obtained from ${\rm{DA}}(P)$ when $(2,4,3,5)$ in implemented, $\eta=((1,s_1), (2,s_3), (3,s_2), (4,s_4), (5,s_5), (6,s_1))$, corrects the inefficiencies generated {\it{after}} the displacement of $3$ from $s_3$. Therefore, in the context of the example above, ${\rm{DA}}(P'_1, P_{-1})$ Pareto dominates the matching obtained when a $\mu$-improving cycle present in $G'$ but not in $G$ is implemented from ${\rm{DA}}(P)$. 

This last property always holds. More formally, let $(i^*_1,\ldots, i^*_r)$ be a $\mu$-improving cycle present in $G'$ (see the proof of Theorem 1, Case I). We affirm that $\mu'$ weakly Pareto dominates under $P$ the matching $\eta$ obtained by implementing $(i^*_1, \ldots, i^*_r)$ from $\mu$. By contradiction, assume that $\eta(k) P_k \mu'(k)$ for some $k\in N$. Since $\mu'$ Pareto dominates $\mu$ under $P$, it follows that $k=i^*_{s}$ for some $s\in \{1,\ldots, r\}$,  $k$ $\mu$-blocks an edge $[i,i^*_{s+1}]$ of $G$, and $\eta(k)=\mu(i^*_{s+1})=\mu'(i)$ [modulo $r$]. Hence, if $s=\mu'(i)$, then $k\succ_s i$ and $s P_k \mu'(k)$. This contradicts the stability of $\mu'$.\\

Kesten (2010) introduces the {\bf{efficiency adjusted deferred acceptance mechanism}} (${\rm{EADAM}}$). This mechanism implements improvement cycles sequentially in such a way that ${\rm{DA}}$ is recalculated after modifying the preferences of some students who agree to waive their priorities at schools. In particular, when everyone agrees to waive their priorities at schools, ${\rm{EADAM}}$ is efficient and Pareto dominates ${\rm{DA}}$. 

Our results are not directly related to Kesten (2010). For instance, the improvement cycle constructed in the proof of \autoref{loc_nonb}---the one that is blocked only by the student who misreports her preferences---does not necessarily coincide with one of the improvement cycles implemented by ${\rm{EADAM}}$. Indeed, it is not difficult to verify that in \autoref{ex-cycles} we have that $${\rm{EADAM}}(P)=((1,s_1), (2,s_2), (3,s_3), (4,s_4), (5,s_5), (6,s_1)).$$  Hence, ${\rm{EADAM}}(P)$ does not implement the $\mu$-improving cycle $(2,4,3,5)$.

\end{document}